\documentclass[a4paper,twocolumn,11pt,accepted=2021-05-05]{quantumarticle}

\pdfoutput=1

\usepackage[numbers, sort&compress, merge]{natbib}

\usepackage[utf8]{inputenc}
\usepackage[english]{babel}
\usepackage[T1]{fontenc}

\usepackage{graphicx}
\usepackage{dcolumn}
\usepackage{bm}
\usepackage{dsfont}
\usepackage{amsmath}
\usepackage{amssymb}
\usepackage{amsthm}
\usepackage[dvipsnames]{xcolor}
\usepackage{chngcntr}
\usepackage{apptools}
\usepackage{placeins}
\AtAppendix{\counterwithin{lemma}{section}}
\usepackage{longtable}
\usepackage{xcolor,colortbl}

\usepackage{hyperref}
\usepackage[capitalise]{cleveref}
\crefname{figure}{Fig.}{Fig.}

\newtheorem{theorem}{Theorem}
\newtheorem{corollary}{Corollary}[theorem]
\newtheorem{lemma}{Lemma}
\theoremstyle{definition}
\newtheorem{definition}{Definition}
\theoremstyle{theorem}

\theoremstyle{definition}

\newcommand{\ham}{H}
\newcommand{\allterms}{\mathcal{S}}
\newcommand{\nonconham}{H_\text{nc}}
\newcommand{\nonconterms}{\mathcal{S}_\text{nc}}
\newcommand{\extraham}{H_\text{c}}
\newcommand{\extraterms}{\mathcal{S}_\text{c}}

\usepackage{mhchem}

\begin{document}

\title{Contextual Subspace Variational Quantum Eigensolver}

\author{William M. Kirby}
\affiliation{Department of Physics and Astronomy, Tufts University, Medford, MA 02155}
\email{william.kirby@tufts.edu}
\orcid{0000-0002-2778-1703}

\author{Andrew Tranter}
\affiliation{Department of Physics and Astronomy, Tufts University, Medford, MA 02155}
\affiliation{Cambridge Quantum Computing, 9a Bridge Street Cambridge, CB2 1UB United Kingdom}
\orcid{0000-0002-0299-3478}

\author{Peter J. Love}
\affiliation{Department of Physics and Astronomy, Tufts University, Medford, MA 02155}
\affiliation{Computational Science Initiative, Brookhaven National Laboratory, Upton, NY 11973}
\orcid{0000-0002-8344-0532}

\begin{abstract}
We describe the \emph{contextual subspace variational quantum eigensolver} (CS-VQE), a hybrid quantum-classical algorithm for approximating the ground state energy of a Hamiltonian. The approximation to the ground state energy is obtained as the sum of two contributions. The first contribution comes from a noncontextual approximation to the Hamiltonian, and is computed classically. The second contribution is obtained by using the variational quantum eigensolver (VQE) technique to compute a contextual correction on a quantum processor. In general the VQE computation of the contextual correction uses fewer qubits and measurements than the VQE computation of the original problem. Varying the number of qubits used for the contextual correction adjusts the quality of the approximation. We simulate CS-VQE on tapered Hamiltonians for small molecules, and find that the number of qubits required to reach chemical accuracy can be reduced by more than a factor of two. The number of terms required to compute the contextual correction can be reduced by more than a factor of ten, without the use of other measurement reduction schemes. This indicates that CS-VQE is a promising approach for eigenvalue computations on noisy intermediate-scale quantum devices.
\end{abstract}

\maketitle

\section{Introduction}
\label{intro}

The \emph{variational quantum eigensolver} (VQE) is the leading algorithm for quantum simulation on noisy intermediate-scale quantum (NISQ) computers, due to the limited resources it requires in both qubit count and coherence time \cite{peruzzo14a,omalley16a,santagati18a,shen2017quantum,paesani2017,kandala17a,hempel18a,dumitrescu18a,colless18a,nam19a,kokail19a,kandala19a,google20a}. VQE is a hybrid quantum-classical algorithm in which the expectation value of the Hamiltonian, or other observable, is computed for an ansatz state generated by a parameterized quantum circuit. Optimization of the ansatz parameters is performed iteratively using an optimization algorithm running on a classical computer. VQE algorithms require a large number of measurements to be performed, and give approximate results due to limitations of the ansatz that can be prepared and noise on the quantum device.  As a result, the largest experiments to date either do not reach chemical accuracy \cite{kandala17a}, do not include all Hamiltonian terms \cite{nam19a}, or simulate restricted models such as Hartree-Fock \cite{google20a}. 

We address these limitations of VQE by providing an approximate simulation method for the full Hamiltonian that can be adjusted to use any amount of available quantum resources. We show that in many cases our method reduces the resources required to reach chemical accuracy. The method is based on the concept of a noncontextual Hamiltonian, in which definite values can be assigned to all Pauli terms simultaneously without contradiction \cite{kirby19a}. In \cite{kirby20a}, we gave a classical algorithm for computing the ground state energies of noncontextual Hamiltonians based on a quasi-quantized model \cite{spekkens07a,spekkens16a}. This raises the question of whether a truly hybrid algorithm can be developed for simulating general Hamiltonians, in which the noncontextual parts of the Hamiltonians are computed classically and contextual quantum corrections are computed using VQE. The method described in this paper is such an algorithm.

The contextual quantum correction is obtained by performing VQE restricted to quantum states that are consistent with the noncontextual ground state. We refer to the space of such quantum states as the \emph{contextual subspace}. The contextual subspace represents the degrees of freedom that remain after the noncontextual degrees of freedom have been fixed. The resulting quantum computations only involve Hamiltonian terms in the complement of the noncontextual Hamiltonian. We also show how to adjust the noncontextual part of the Hamiltonian, in order to move more of the computation onto the quantum computer, while preserving the structure of the quasi-quantized model. 
The technique for accomplishing this is related to ``subspace-search VQE'' \cite{nakanishi19a}, in which excited energies are found by restricting the search space to be orthogonal to the (previously approximated) ground state.
In our case, we are not looking for excited states, but we are implementing VQE in a restricted search space, and part of the technique for achieving this is similar to that in \cite{nakanishi19a} (see \cref{quantum_part} for details of the technique).

As an example, suppose we want to apply VQE to a Hamiltonian on $n$ qubits, but the available quantum processor has only $q$ qubits. CS-VQE permits us to adjust the noncontextual approximation method so that the associated quantum correction uses exactly $q$ qubits. Increasing the number of qubits used on the quantum processor monotonically improves the quality of the overall approximation, interpolating between the noncontextual approximation with no quantum correction and full VQE. Thus, we can tune the quantum part to fit the available quantum resources, with the classical method making up the difference.

CS-VQE is the first hybrid quantum-classical algorithm of its kind, where a nonclassicality criterion (in our case, contextuality) is used to isolate the intrinsically quantum part of a quantum algorithm, and the classical remainder of the algorithm is simulated classically.
Our prior works \cite{kirby19a,kirby20a} defined contextual and noncontextual Hamiltonians, and gave a classical model for noncontextual Hamiltonians, respectively, so the current algorithm is distinct because it is applicable to arbitrary Hamiltonians and because it has a quantum component as well as a classical component.

In the remainder of the introduction, we review the necessary information from \cite{kirby20a} about classical simulation of noncontextual Hamiltonians. In \cref{quantum_part}, we describe the quantum correction procedure.
In \cref{c_subspace_vqe}, we describe how to implement CS-VQE, and show the results of simulating CS-VQE on Hamiltonians for small molecules. Finally, in \cref{conclusion}, we summarize our results and discuss their implications for NISQ computing.

\subsection{Noncontextual Pauli Hamiltonians}
\label{noncontextual_part}

A set of observables is \emph{noncontextual} when it is possible to assign values to them simultaneously without contradiction \cite{bell64a,bell66a,kochen67a,spekkens05a,abramsky11a,raussendorf13a,howard14a,cabello14a,cabello15a,ramanathan14a,de_silva17a,amaral17a,xu18a,raussendorf19c,duarte18a,mansfield18a,okay18a,raussendorf19b}. How can we determine if a set $\allterms$ of Pauli operators is noncontextual? First remove the subset $\mathcal{Z}$ of terms that commute with the whole of $\allterms$. The set $\allterms$ is noncontextual if and only if commutation is an equivalence relation on $\allterms\setminus\mathcal{Z}$~\cite{kirby19a}. The equivalence classes of commutation in $\allterms\setminus\mathcal{Z}$ are {\em cliques} $C_i$ for $i=1,2,...,N$. Operators in the same clique commute, while operators in different cliques anticommute.

Let $\allterms$ be the set of Pauli terms in a general Hamiltonian $H$. We divide $\allterms$ into a noncontextual subset $\nonconterms$ and its complement $\extraterms$. This induces a decomposition of $H$ into a noncontextual part $\nonconham$ whose Pauli terms are $\nonconterms$, and $\extraham$ whose Pauli terms are $\extraterms$~\cite{kirby20a}:
\begin{equation}
\label{ham_decomp}
  \ham = \nonconham+\extraham.
\end{equation}
We also require that $\nonconterms$ be \emph{closed under inference within $\allterms$}:
\begin{definition}
\label{closure_under_inference_within}
	$\nonconterms$ is \emph{closed under inference within $\allterms$} if any operators in $\allterms$ whose values can be inferred from the values of operators in $\nonconterms$ must be included in $\nonconterms$~\cite{kirby20a}.
\end{definition}
Closure under inference is reviewed in detail in \cref{noncontextual_part_app}.
The decomposition \eqref{ham_decomp} is the basis of the CS-VQE method. We will obtain an efficient classical description of the eigenspaces of $\nonconham$, and use this and $\extraham$ to quantum compute a correction to the noncontextual approximation to the ground-state energy.

Because all terms in $\nonconham$ can simultaneously be assigned definite values without contradiction we can introduce a phase space description of its eigenspaces~\cite{kirby19a,raussendorf19b}. The phase space points are the possible joint value assignments to a set of observables derived from $\nonconterms$, which we describe below. The eigenstates of $\nonconham$ are probability distributions over this phase space. This is a quasi-quantized model: a classical phase-space model with an uncertainty relation imposed upon the allowed probability distributions (sometimes called \emph{epistemic states}) on the phase space~\cite{spekkens07a,spekkens16a}. We refer the reader to~\cite{kirby20a} and \cref{noncontextual_part_app} for further general points about quasi-quantized models.

We now describe the states of the model, which we call \emph{noncontextual states}. We first identify a set of observables that define the phase-space points in the model:
\begin{equation}
\label{noncon_generators_points}
    G\cup\{A_1,A_2,...,A_N\}.
\end{equation}
Each $A_i$ is one element of the corresponding clique $C_i$, so the $A_i$ pairwise anticommute. $G$ is an independent generating set for the Abelian group $\overline{\mathcal{Z}}$, which includes $\mathcal{Z}$ as well as all products of pairs of operators in the same clique. The phase space points are all assignments of values $\pm1$ to the observables in the set \eqref{noncon_generators_points}~\cite{kirby20a}.

The noncontextual states are probability distributions over the phase space points generated by~\eqref{noncon_generators_points}. Probability distributions corresponding to valid quantum states must obey an uncertainty relation~\cite{spekkens07a,spekkens16a}. A sufficient condition is that the commuting generators $G_j\in G$ take definite values, and that the expectation values of the $A_i$ form a unit vector~\cite{kirby20a}. This means that each noncontextual state is defined by parameters $(\vec{q},\vec{r})$ such that
\begin{equation}
\label{noncon_params_def}
	\langle G_j\rangle=q_j=\pm1,\quad\langle A_i\rangle=r_i,\quad|\vec{r}|=1.
\end{equation}
In \cite{kirby20a}, we showed how these expectation values for the set \eqref{noncon_generators_points} induce expectation values for all terms $\nonconterms$ in the noncontextual part $\nonconham$ of the Hamiltonian; consequently, a noncontextual state induces an expectation value for $\nonconham$.
We also proved that all expectation values for $\nonconham$ can be generated in this way.
Minimizing this expectation value by varying the noncontextual state $(\vec{q},\vec{r})$ thus provides a variational estimate of the ground state energy of $\nonconham$~\cite{kirby20a}.
We refer to the minimizing assignment $(\vec{q},\vec{r})$ as the \emph{noncontextual ground state}.

The $A_i$ are anticommuting Pauli operators and $\vec{r}$ is a real unit vector, so the observable
\begin{equation}
\label{A_def}
    \mathcal{A}\equiv\sum_{i=1}^Nr_iA_i,
\end{equation}
is a rotated Pauli operator, and thus has eigenvalues $\pm1$. The unitary that maps $\mathcal{A}$ to a single Pauli operator is a sequence of $N-1$ rotations generated by Pauli operators, all of which preserve the $G_j$.
If each of the $A_i$ has expectation value $r_i$ then $\mathcal{A}$ has expectation value $+1$, and vice versa~\cite{kirby20a}.
Thus, the noncontextual state with parameters $(\vec{q},\vec{r})$ is equivalent to a joint value assignment for the set of observables
\begin{equation}
\label{noncon_generators}
    G\cup\{\mathcal{A}\},
\end{equation}
where the value assignments are $G_j\mapsto q_j=\pm1$ for each $j$ and $\mathcal{A}\mapsto+1$.
We refer to the observables in \eqref{noncon_generators} as the \emph{noncontextual generators}.

The noncontextual states therefore correspond to subspaces of quantum states that are stabilized by the operators $q_jG_j$ and by $\mathcal{A}$. These are almost stabilizer subspaces in the usual sense (see e.g. \cite[Sec. 10.3]{nielsen01}), except that $\mathcal{A}$ is not a single Pauli operator, but is unitarily equivalent to one. Therefore, a noncontextual state can be thought of as a stabilizer subspace, one of whose stabilizers has been rotated by an efficiently-describable unitary.

\section{Quantum correction}
\label{quantum_part}

Let $(\vec{q},\vec{r})$ be the noncontextual ground state. If we take the resulting energy of $\nonconham$ as a classical estimate of the ground state energy of the full Hamiltonian $\ham$, we can obtain a quantum correction by minimizing the energy of the remaining terms in the Hamiltonian over the quantum states that are consistent with the noncontextual ground state.
As discussed above, this common eigenspace is a stabilizer subspace up to a rotation on one of the stabilizers.
We refer to this subspace as the \emph{contextual subspace}.

Before we discuss how to find quantum corrections, we establish when such corrections can appear:
\begin{theorem}
\label{null_exp_thm}
    Let $\allterms$ be a set of Pauli operators, and let $\nonconterms$ be a noncontextual subset that is closed under inference within $\allterms$ (see \cref{closure_under_inference_within}).
    Then for any noncontextual state $(\vec{q},\vec{r})$ as in \eqref{noncon_params_def} describing $\nonconterms$, there exists a quantum state consistent with $(\vec{q},\vec{r})$ (i.e., that gives the same expectation values for $\nonconterms$ as $(\vec{q},\vec{r})$) for which the expectation value of every operator in $\extraterms\equiv\allterms\setminus\nonconterms$ is zero.
\end{theorem}
\noindent
The proof may be found in \cref{proofs}, and follows from the fact that $\nonconterms$ is closed under inference: no value of any operator in $\extraterms$ can be inferred from the values of operators in the noncontextual part of the Hamiltonian \cite{kirby20a}.
\cref{null_exp_thm} has two useful corollaries.
\begin{corollary}
\label{no_qc_no_exp} 
    If the noncontextual states of $\nonconham$ uniquely identify quantum states, then for any noncontextual state the expectation value of every term in $\extraham$ is zero, i.e., no quantum correction is possible.
\end{corollary}
\begin{proof}
    By \cref{null_exp_thm}, there exists a quantum state $|\psi\rangle$ consistent with the noncontextual state for which the expectation value of every term not in the noncontextual part is zero.
    Thus if the noncontextual state uniquely identifies a quantum state, it must be $|\psi\rangle$.
\end{proof}

\begin{corollary}
\label{noncon_variational}
    As an approximation method for a general Hamiltonian, a quasi-quantized model of the noncontextual part is variational, i.e., the resulting approximate ground state energies are lower-bounded by the true ground state energy.
    It remains variational when the quantum correction is included.
\end{corollary}
\begin{proof}
    From \cref{null_exp_thm}, it follows that there exists a ground state of the noncontextual part for which the expectation value of every other term in the Hamiltonian is zero.
    Hence, the ground state energy of the noncontextual part is a possible expectation value for the energy of the full Hamiltonian, so it is lower-bounded by the ground state energy of the full Hamiltonian.
    To compute the quantum correction we minimize the expectation value of the full Hamiltonian over quantum states consistent with a given noncontextual state. This produces a variational estimate of the energy with the contribution from $\nonconham$ given by the noncontextual state. 
\end{proof}

\subsection{Mapping a contextual subspace to a stabilizer subspace}
\label{diag_map}

We now show how to map the contextual subspace corresponding to the noncontextual ground state to a subspace stabilized by single-qubit $Z$ operators. To achieve this goal we rotate the $G_j$ and subsequently the single operator $\mathcal{A}$ to single-qubit $Z$ operators.

The number of $G_j$ is some $M<n$, as discussed in \cite{kirby20a}. Therefore, the $G_j$ can be mapped to single-qubit $Z$ operators by a sequence of at most $2M$ $\frac{\pi}{2}$-rotations\footnote{See for example the solution to problem 3 in \cite{yan12a}.} (for completeness, we provide a constructive proof as \cref{diagonalizecommutinglemma} in \cref{proofs}). 
These rotations are Clifford operators, so they map the remaining Pauli operators in the Hamiltonian back to single Pauli operators while preserving their commutation relations.
Let $D$ denote the composition of these rotations, and let $H'\equiv DHD^\dagger$ be the rotated Hamiltonian.
After applying $D$, the noncontextual generators $G_j$ have been mapped to single-qubit $Z$ operators $G_j'$.
Without loss of generality, since we have already found the noncontextual ground state at this point we can choose the signs of the $G_j'$ such that they all have eigenvalue $+1$ in the noncontextual ground state, and thus stabilize it.
We refer to this basis as the ``rotated basis.''

Once $D$ has been applied, we map $\mathcal{A}'=D\mathcal{A}D^\dagger$ to a single-qubit $Z$ operator as well. $\mathcal{A}'$ is a normalized linear combination of the anticommuting Pauli operators $A_i'=DA_iD^\dagger$, as in \eqref{A_def}, so we use the sequence of rotations employed in unitary partitioning~\cite{izmaylov19a,zhao20a}. The result is a sequence of $N-1$ rotations generated by 
products of pairs of the $A_i'$; we denote by $R$ the composition of these rotations, and the result of applying it to $\mathcal{A}'$ is
\begin{equation}
    R^\dagger\mathcal{A}'R=A_1',
\end{equation}
where $A_1'$ is a single Pauli operator.

The rotations forming $R$ are generated by products of the $A_i'$, and the $A_i'$ commute with the operators $G_j'$, so $R$ commutes with the operators $G_j'$. Thus $A_1'$ commutes with and is independent of the operators in $G'$, so since it is a single Pauli operator, we can use \cref{diagonalizecommutinglemma} to map it to a single-qubit $Z$ operator as well, without disturbing the operators in $G'$. Let $D_{\mathcal{A}'}$ denote the full rotation that maps $\mathcal{A}'$ to a single-qubit $Z$ operator $A''$.

\subsection{Restricting the Hamiltonian to a contextual subspace}
\label{eigenspace_search}

In the rotated basis, we will restrict the Hamiltonian to the subspace stabilized by the noncontextual generators $G_j'$.
We will then obtain the quantum correction by minimizing the expectation value of this restricted Hamiltonian over $+1$-eigenvectors of the remaining noncontextual generator $\mathcal{A}'$.

Let $\mathcal{H}_1$ denote the Hilbert space of the $n_1$ qubits acted on by the single-qubit Pauli $Z$ operators $G_j'$, and let $\mathcal{H}_2$ denote the Hilbert space of the remaining $n_2$ qubits. Thus the full Hilbert space is $\mathcal{H}=\mathcal{H}_1\otimes\mathcal{H}_2$ and the total number of qubits is $n=n_1+n_2$. 
The contextual part of the Hamiltonian in the rotated basis is:
\begin{equation}\label{rbextra}
    \extraham'=\sum_{P\in\extraterms'}h_PP.
\end{equation}
The set of Pauli terms in $\extraham'$ is $\extraterms'$ and terms in $\extraterms'$ in general act on both of the subspaces $\mathcal{H}_1$ and $\mathcal{H}_2$.

We can write the terms $P$ in~\eqref{rbextra} as
\begin{equation}
    P=P_1\otimes P_2,
\end{equation}
where $P_1$ is a Pauli operator acting on $\mathcal{H}_1$ and $P_2$ is a Pauli operator acting on $\mathcal{H}_2$. $P$ commutes with an element of $G'$ if and only if $P_1\otimes\mathds{1}_{\mathcal{H}_2}$ does (where $\mathds{1}_{\mathcal{H}_2}$ denotes the identity operator acting on $\mathcal{H}_2$), since the operators in $G'$ act only on $\mathcal{H}_1$. Since the noncontextual state corresponds to a subspace stabilized by $G'$, if $P$ anticommutes with any element of $G'$, then its expectation value in the noncontextual state is zero. Hence, any $P$ that admits a quantum correction must commute with all elements of $G'$, and thus $P_1\otimes\mathds{1}_{\mathcal{H}_2}$ must as well. The elements of $G'$ comprise all single-qubit $Z$ operators acting in $\mathcal{H}_1$, so $P_1$ must be a product of such operators. Hence the expectation value of $P_1$ is some $p_1=\pm1$, determined by the noncontextual state.

Let $|\psi_{(\vec{q},\vec{r})}\rangle$ be any quantum state consistent with the noncontextual state $(\vec{q},\vec{r})$. The action of any term $P$ that admits a quantum correction on $|\psi_{(\vec{q},\vec{r})}\rangle$ has the form
\begin{equation}
\begin{split}
    P|\psi_{(\vec{q},\vec{r})}\rangle&=\big(P_1\otimes P_2\big)|\psi_{(\vec{q},\vec{r})}\rangle\\
    &=p_1\big(\mathds{1}_{\mathcal{H}_1}\otimes P_2\big)|\psi_{(\vec{q},\vec{r})}\rangle.
\end{split}
\end{equation}
Therefore, if we denote by $\extraham'|_{(\vec{q},\vec{r})}$ the restriction of $\extraham'$ to its action on the noncontextual ground state $(\vec{q},\vec{r})$,
\begin{equation}
\begin{split}
    \extraham'|_{(\vec{q},\vec{r})}&=\sum_{\substack{P\in\extraterms'\\\text{s.t.}~[P,G'_i]=0\\\forall G'_i\in G'}}p_1 h_P\mathds{1}_{\mathcal{H}_1}\otimes P_2\\
    &=\mathds{1}_{\mathcal{H}_1}\otimes\extraham'|_{\mathcal{H}_2}\color{white}{\bigg)}
\end{split}
\end{equation}
for
\begin{equation}
    \extraham'|_{\mathcal{H}_2}\equiv\sum_{\substack{P\in\extraterms'\\\text{s.t.}~[P',G'_i]=0\\
    \forall G'_i\in G'}}p_1h_PP_2.
\end{equation}
This is a Hamiltonian on $n_2$ qubits, for $n_2$ given by
\begin{equation}
\label{n_qubits_qc}
    n_2=n-|G|,
\end{equation}
where $|G|$ is the number of noncontextual generators $G_j$.
Furthermore, if terms $P\in\extraterms'$ are distinct only on their tensor factors $P_1$, the remaining operators $P_2$ in $\extraham'|_{\mathcal{H}_2}$ will be identical. Also, any terms that anticommute with any of the noncontextual generators $G_j$ are dropped entirely (since their expectation values are zero). Thus, the restricted Hamiltonian $\extraham'|_{\mathcal{H}_2}$ may contain fewer than $|\extraterms'|$ terms. This is illustrated in \cref{fig:terms_csvqe_vs_full_practical}, in \cref{cssvqe_examples}.

\subsection{Optimizing within a noncontextual subspace}

To obtain the quantum correction, we perform $n_2$-qubit VQE on the restricted Hamiltonian $\extraham'|_{\mathcal{H}_2}$ within the contextual subspace. The contextual subspace is the subspace of $\mathcal{H}_2$ that forms the $+1$-eigenspace of $\mathcal{A}'$, the remaining noncontextual generator in the rotated basis.
To search within the $+1$-eigenspace of $\mathcal{A}'$, we can prepare ans\"atze in the $+1$-eigenspace of $A''$ (a single-qubit $Z$ operator), and then apply the inverse of $D_{\mathcal{A}'}$.
This guarantees that every ansatz state is consistent with the noncontextual ground state $(\vec{q},\vec{r})$.

Note that we only explicitly restrict the rotated Hamiltonian to the subspace stabilized by $G_j'$, whereas we restrict to the $+1$-eigenbasis of $\mathcal{A}'$ at the level of the ansatz. This is because although the operation $D_{\mathcal{A}'}$ that diagonalizes $\mathcal{A}'$ can be efficiently implemented on a quantum computer, it is not a Clifford operation.  Conjugating the Hamiltonian with $D_{\mathcal{A}'}$ can increase the number of terms by a factor of $\Theta(2^N)$, where $N$ is the number of cliques. Thus if $N$ is small, we could classically compute the Hamiltonian restricted to the $+1$ eigenbasis of $\mathcal{A}'$ and then perform VQE on this Hamiltonian. This would save one qubit and permit an unconstrained search for the quantum correction, but since $N$ can in principle scale as $\Theta(n)$ this approach will not be efficient in general.

\subsection{Example}
\label{examples}

As an example, we construct a Hamiltonian for which most of the terms are included in the noncontextual part. Let $\allterms=\nonconterms\cup\extraterms$, where
\begin{equation}
\label{qc_example_terms}
\begin{split}
    \nonconterms = \{ZII,&IXI,IYI,IZX,IZY,IZZ,\\
    &ZXI,ZYI,ZZX,ZZY,ZZZ\},\\
    \extraterms = \{IIX,&IIY,IIZ\}.
\end{split}
\end{equation}
The set of terms $\nonconterms$ is noncontextual, partitioning into $\mathcal{Z}=\{ZII\}$ (recall that $\mathcal{Z}$ is the set of terms that commute with all others), and five cliques, $\{IXI,ZXI\}$, $\{IYI,ZYI\}$, $\{IZX,ZZX\}$, $\{IZY,ZZY\}$, and $\{IZZ,ZZZ\}$.
Thus we may choose
\begin{equation}
\begin{split}
    &A_1=IXI,~A_2=IYI,~A_3=IZX,\\
    &A_4=IZY,~A_5=IZZ.
\end{split}
\end{equation}
The extra terms $\extraterms$ all commute with $\mathcal{Z}$.
In this case, $G=\mathcal{Z}$ since $\mathcal{Z}$ contains only one operator, and this operator is already a single-qubit $Z$ operator, so $D$ is the identity.

Thus $\mathcal{H}_2$ is the Hilbert space of the second two qubits, so for
\begin{equation}
    \extraham'=\extraham=h_{IIX}IIX+h_{IIY}IIY+h_{IIZ}IIZ
\end{equation}
for some coefficients $h_{IIX},h_{IIY},h_{IIZ}$, the restriction to $\mathcal{H}_2$ is
\begin{equation}
    \extraham'|_{\mathcal{H}_2}=\extraham=h_{IIX}IX+h_{IIY}IY+h_{IIZ}IZ.
\end{equation}
We also have
\begin{equation}
    \mathcal{A}'=\mathcal{A}=r_1A_1+r_2A_2+r_3A_3+r_4A_4+r_5A_5
\end{equation}
for some unit vector $\vec{r}$; the restriction of $\mathcal{A}'$ to $\mathcal{H}_2$ is thus
\begin{equation}
    \mathcal{A}'|_{\mathcal{H}_2}=r_1XI+r_2YI+r_3ZX+r_4ZY+r_5ZZ,
\end{equation}
so $D_{\mathcal{A}'}$ is the rotation that maps this to a single-qubit $Z$ operator, as described in \cref{diag_map}.
We can choose
\begin{equation}
    D_{\mathcal{A}'}\mathcal{A}'|_2D_{\mathcal{A}'}^\dagger=ZI;
\end{equation}
in this case, for an ansatz we may prepare any state whose value is $|0\rangle$ for the first qubit in $\mathcal{H}_2$, and then apply $D_{\mathcal{A}'}^\dagger$ to this state.

Thus, we reduce an initial Hamiltonian on three qubits to a noncontextual approximation and a quantum correction that may be implemented on a two-qubit quantum processor.

To evaluate the performance of the resulting approximations, we generated 10000 Hamiltonians with the terms \eqref{qc_example_terms} by choosing coefficients for them uniformly at random from $[-1,1]$.
The resulting fractional errors in the ground state energies are plotted in \cref{qc_example_ground_states}; the average fractional error is 0.257 for the noncontextual approximation alone, and 0.0268 when the quantum correction is included.
The quantum corrections were simulated classically by directly evaluating the lowest eigenvalues of the Hamiltonians restricted to the noncontextual ground states.

\begin{figure}
\centering
\includegraphics[width=\columnwidth]{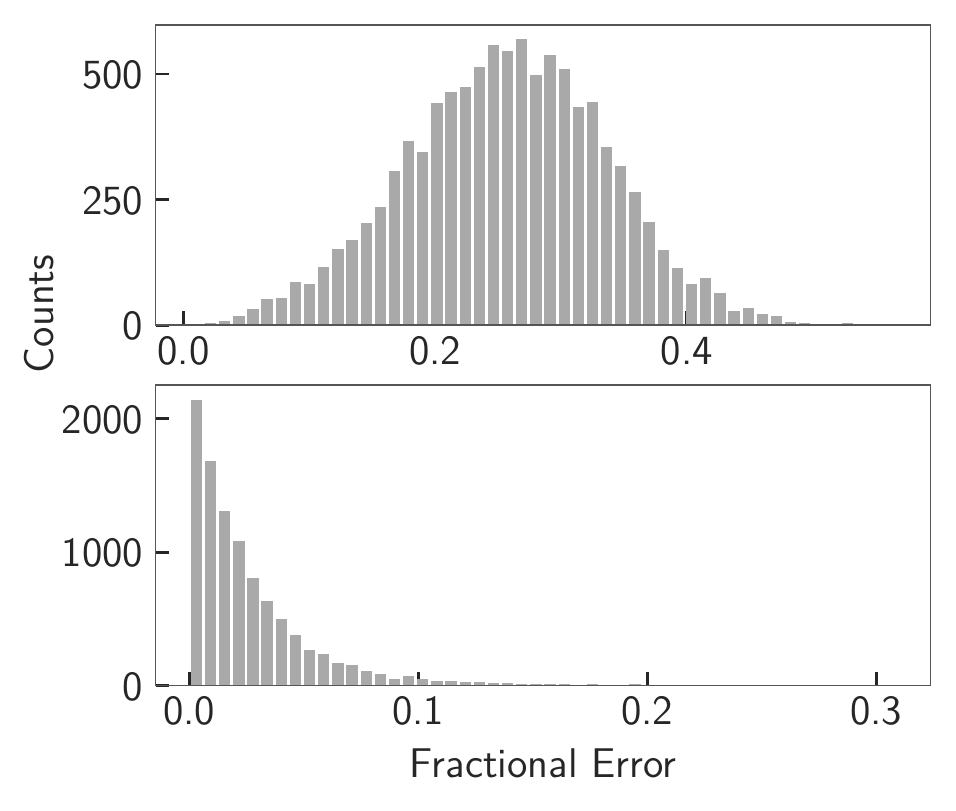}
\caption{Comparison of fractional errors in the noncontextual approximation of the ground state energy (upper plot), and in the noncontextual approximation with quantum correction (lower plot). The histogram points were generated by $10000$ Hamiltonians with terms \eqref{qc_example_terms} and uniformly random coefficients in $[-1,1]$. The mean fractional error without quantum correction is $0.257$, and the mean fractional error with quantum correction is $0.0268$.\label{qc_example_ground_states}}
\end{figure}

\section{Contextual subspace VQE}
\label{c_subspace_vqe}

The quantum correction to noncontextual approximations discussed in \cref{quantum_part} allows us to use limited quantum resources to improve a classical simulation result.
In this section we explain how we can systematically step back from the original noncontextual approximation in order to enlarge the contextual subspace, thus improving the overall accuracy of the approximation by using more quantum resources.
This provides a parameter that can be specified based on the quantum resources available, taking us from the optimal noncontextual approximation at one extreme to full VQE at the other.
We call this method \emph{contextual subspace VQE}.

\subsection{Method}
\label{CSVQE_method}

We begin with a Hamiltonian $\ham$ whose noncontextual approximation is $\nonconham$. As discussed above, the noncontextual ground state corresponds to a joint eigenspace of the noncontextual generators $G\cup\{\mathcal{A}\}$. We can trade accuracy of the noncontextual approximation for an improved quantum correction by decreasing the size of $G\cup\{\mathcal{A}\}$, which increases the dimension of the contextual subspaces. We accomplish this by decreasing the size of $G$, the set of generators for the commuting part of $\nonconham$. Since the number of qubits used in the quantum correction procedure is the total number of qubits minus the number of generators in $G$ (see \eqref{n_qubits_qc}), reducing the size of $G$ increases the dimension of the search space for the quantum processor.

We work in the rotated basis, as in \cref{quantum_part}. In this basis, we select some subset of the noncontextual generators $G_j'$, and remove all terms generated by them from the noncontextual part. Since the $G_j'$ are single-qubit $Z$ operators, this means that for each generator to be dropped we remove from the noncontextual part all terms containing the corresponding single-qubit $Z$ operator as a tensor factor. All the terms thus removed should be added to the quantum correction Hamiltonian $\extraham'$ (as in \cref{eigenspace_search}). We now implement the quantum correction on this expanded $\extraham'$, keeping the same noncontextual ground state that we began with, but only applying its value assignments to the generators that remain in the noncontextual part.

The new noncontextual approximation on its own will in general be worse than the original noncontextual approximation. However, after including the new quantum correction the overall approximation cannot be worse, and will in general be better.
This is because the values assigned in the original noncontextual approximation and quantum correction are still consistent with the noncontextual ground state, so quantum states that obtain those values are included in the new quantum search space.
Thus in the worst case the new approximation will only recover the original approximation.
If the new quantum correction is nonzero for any additional terms, the new approximation will be strictly better than the original approximation.
In the limit where we remove all terms from the noncontextual part and simulate them on the quantum computer, there will be no noncontextual approximation left, and we will have recovered full VQE.

The additional terms that can have nonzero quantum corrections after the removal procedure are those that anticommute with any of the generators $G_j'$ that were removed, but commute with the remaining generators.
These terms were previously restricted to null expectation values only because the noncontextual state was required to be a joint eigenstate of the removed generators, so when that is no longer enforced their expectation values can vary.
Therefore, we can choose which subset of the generators to remove based on which will permit the optimal quantum correction.

Note that classically simulating the noncontextual part of the Hamiltonian is NP-complete, so in worst cases the classical simulation part of CS-VQE will not perform well \cite{kirby20a}.
However, worst case Hamiltonians for standard VQE are QMA-complete, meaning that a similar argument applies to VQE in general.
Hence, in both cases we are interested in heuristic performance for specific Hamiltonians of interest, rather than worst cases.
Framed in this way, what CS-VQE does is take standard VQE, which is a heuristic for an optimization problem over a set of parameters for a quantum circuit, and transform it into two heuristics (one classical and one quantum) for two smaller optimization problems.
In practice, we have found that a combination of Monte-Carlo and gradient descent methods works well for the classical part of the algorithm, but continuing to optimize this is a topic for future work.

\subsection{Applications}
\label{cssvqe_examples}

\begin{figure*}
\centering
\includegraphics[width=7in]{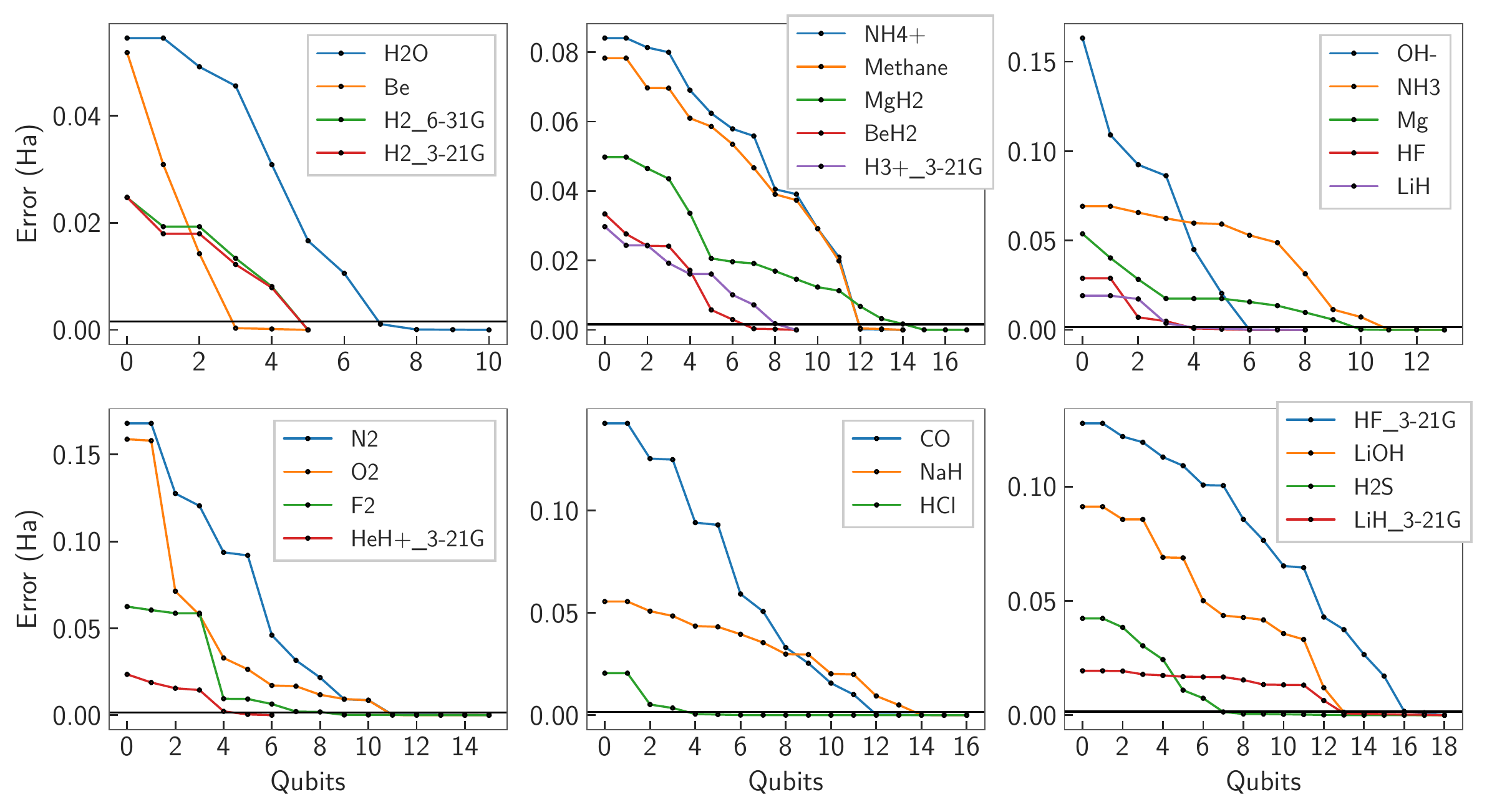}
\caption{
CS-VQE approximation errors versus number of qubits used on the quantum computer, for tapered molecular Hamiltonians. All Hamiltonians whose curves overlap in the region below chemical accuracy have the same total numbers of qubits. The solid black lines indicate chemical accuracy. Within each subplot, the ordering of the legend matches the vertical ordering of the leftmost points in the curves.
\label{errors_vs_qubits_pract}
}
\end{figure*}

\begin{figure}
\centering
\includegraphics[width=\columnwidth]{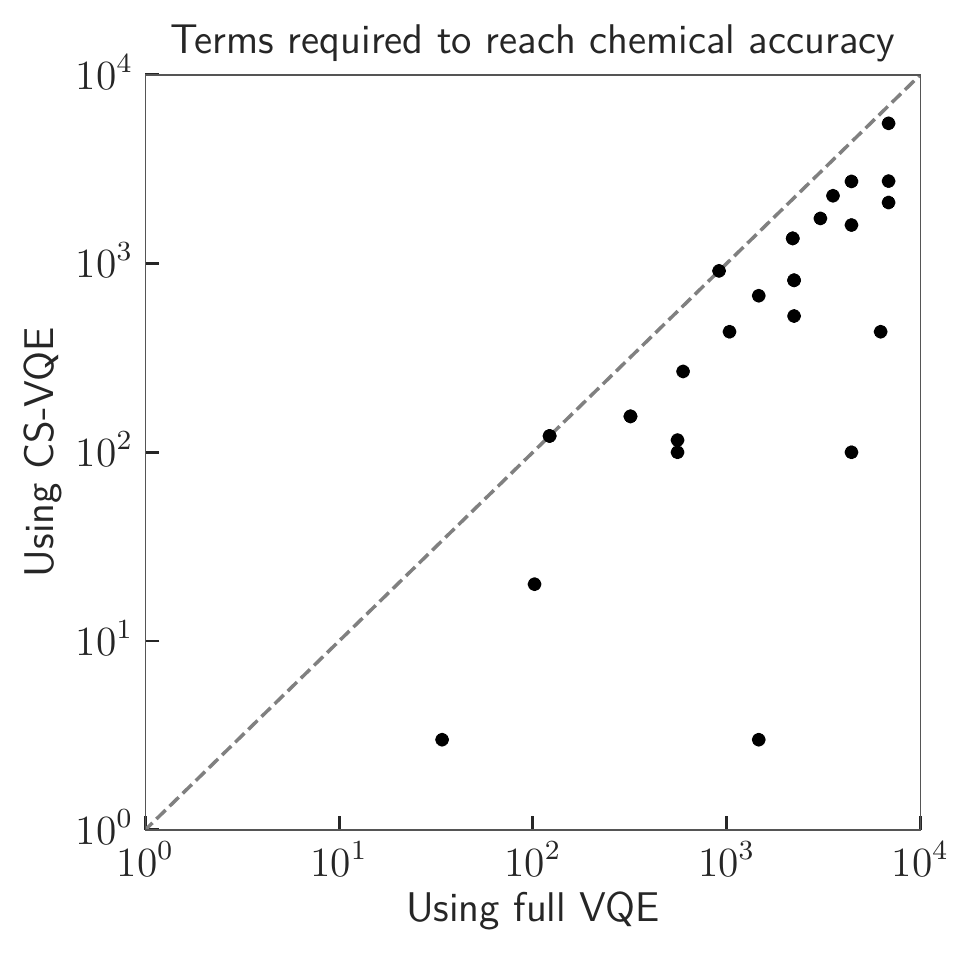}
\caption{
Number of terms simulated on the quantum processor required to reach chemical accuracy using CS-VQE versus using full VQE.
The dashed line marks equality. All points represent either one, two, or three Hamiltonians.
\label{fig:terms_csvqe_vs_full_practical}}
\end{figure}

We tested CS-VQE on a set of electronic structure Hamiltonians in the Jordan-Wigner mapping \cite{jordan28a}. In order to distinguish CS-VQE from qubit tapering, we first tapered the Hamiltonians using symmetries as in \cite{bravyi17a,setia20a}, then implemented CS-VQE in order to remove even more qubits. The initial noncontextual approximation Hamiltonians were chosen via a greedy classical algorithm, as described in \cite{kirby20a}. This algorithm runs in $O(N^5)$ time for an $N$-term Hamiltonian, since testing a particular Hamiltonian for noncontextuality takes $O(N^3)$ time \cite{kirby19a}, and a greedy algorithm that adds optimal terms one at a time requires $O(N^2)$ steps.
This method is not optimal, but is efficient.
The quantum parts of the procedures were simulated classically by directly evaluating the lowest eigenvalues of the quantum correction Hamiltonians. The results are given in \cref{errors_vs_qubits_pract}, which shows the overall CS-VQE approximation errors versus the number of qubits used on the quantum computer, and in \cref{fig:terms_csvqe_vs_full_practical}, which shows the number of terms that must be simulated on the quantum computer in order to reach chemical accuracy using CS-VQE.
Our code is available on GitHub\footnote{Source code: \url{https://github.com/wmkirby1/ContextualSubspaceVQE}}, and may be used to reproduce our results or to apply CS-VQE to new Hamiltonians of the reader's choosing.

As noted at the end of \cref{CSVQE_method}, CS-VQE is sensitive to the order in which the qubits are moved from the noncontextual approximation to the quantum processor. In the calculations to obtain Figs.~\ref{errors_vs_qubits_pract} and \ref{fig:terms_csvqe_vs_full_practical}, we used a heuristic that begins with the noncontextual approximation, then adds qubits to the quantum correction two at a time, greedily choosing each pair to maximize the decrease in the ground state energy estimate.
This method is informed by the structure of the noncontextual and contextual parts of the molecular Hamiltonians, and performed best out of the heuristics we tried that can be implemented efficiently without performing full VQE.
Details of the implementation of the heuristic are given in \cref{applications_app}.

This heuristic involves running CS-VQE repeatedly, since for sufficiently large applications one would have to use the quantum processor to compute the quantum corrections on the way to choosing the set of qubits for the final quantum correction. However, these preliminary computations would only be necessary once the number of qubits chosen becomes unfeasible for classical simulation, and from \cref{fig:terms_csvqe_vs_full_practical} we see that the number of terms required to reach chemical accuracy can be many times smaller than the number of terms required to implement full VQE.
Therefore, even the repeated runs of CS-VQE required for this heuristic can require fewer measurements overall than full VQE, and of course they also require fewer qubits.

Alternatively, one could use a heuristic to determine the order without evaluating energies at all. However, all variants of this that we tried had substantially worse performance than the heuristic discussed above, so we suggest using that heuristic for real applications. We also tested an inefficient ``optimal'' heuristic that begins from full VQE and moves qubits to the noncontextual approximation one at a time, greedily minimizing the error penalty for each.
Actually implementing this heuristic is even more costly than full VQE, but it did identify the optimal qubit orderings in cases small enough for us to find these by brute-force search. The first heuristic discussed above performed nearly as well as the ``optimal'' heuristic, requiring the same number of qubits to reach chemical accuracy in most cases and only one extra in the remaining few cases. Details of all of the heuristics and their relative performance are discussed in \cref{applications_app}.

\section{Conclusion}
\label{conclusion}

In this paper, we showed how to use a quantum computer to obtain a correction to a noncontextual approximation of a ground state energy. We then showed how to adjust the number of qubits used on the quantum computer in order to increase the accuracy of the hybrid approximation. This method, contextual subspace VQE or CS-VQE, is a true hybrid quantum-classical algorithm, in which the quantum resources used may be set to match whatever resources are available, and the classical approximation algorithm accounts for the remainder. The method is approximate, but variational, as is VQE itself. Exact methods will only achieve approximate results on NISQ devices due to their noisy character. CS-VQE allows the quantum resources used to be increased systematically until the desired precision is achieved, if possible.

Standard VQE is a heuristic algorithm: there are no analytic characterizations of its performance for general Hamiltonians, or even for special classes like electronic structure Hamiltonians, upon scaling the system size.
This is also true for CS-VQE: its performance is sensitive to the specific problem to which it is applied.
We do not analytically characterize the errors as a function of the number of qubits used on the quantum processor.
However, the examples in \cref{cssvqe_examples} illustrate that CS-VQE performs well in many cases of interest going well beyond the scale of VQE implementations to date, so we hope that as the available quantum processors continue to grow, CS-VQE can be used to allow larger systems to be simulated using those processors.

The technique for restricting the quantum correction to the subspace consistent with the noncontextual ground state may appear reminiscent of using qubit tapering to exploit symmetries as described in \cite{bravyi17a,setia20a}.
However, in CS-VQE the symmetries are intrinsic to the noncontextual ground state, rather than to the Hamiltonian (as in \cite{bravyi17a,setia20a}), and are thus under the experimenter's control. We illustrated this point in \cref{cssvqe_examples} by applying CS-VQE to Hamiltonians that were already tapered using the methods of \cite{bravyi17a,setia20a}; using CS-VQE we can eliminate additional qubits at will.

In this paper we did not explore how to implement ans\"atze for the restricted VQE instance used by CS-VQE, instead finding the exact ground states of the contextual parts.
However, standard ansatz classes, like unitary coupled-cluster (UCC) for electronic structure Hamiltonians \cite{yung14a,mcclean16a,romero18a}, can be transformed into ans\"atze for CS-VQE by projecting the gates onto the contextual subspace, just as the contextual part of the Hamiltonian is restricted to the contextual subspace.
Detailed study of this is a topic for future research.

One concern for standard VQE as well as for CS-VQE is that the ansatz may suffer from the \emph{barren plateau problem} \cite{mcclean18a,uvarov21a,cerezo21a}, where the gradient of the cost function (in this case expected energy) vanishes exponentially with the system size.
It is hoped that for standard VQE, using physically-motivated ans\"atze like UCC may avoid the barren plateau problem, so since we can use projections of the same ans\"atze for CS-VQE, this same hope transfers to our case.
However, even physically-motivated ans\"atze may be subject to \emph{noise-induced barren plateaus} \cite{wang20a}: to the best of our knowledge, all variational quantum algorithms have the potential to fail in this way, including CS-VQE.
Nonlinear optimization and its attendant problems, including barren plateaus, may be avoided by the use of quantum imaginary time evolution (QITE) or similar methods \cite{motta20a,mcardle19a}. In our case, QITE could be applied directly to the contextual part of the Hamiltonian.

It is possible that some of the qubit and term reductions we obtained using CS-VQE have explanations in terms of chemistry.
However, in such cases CS-VQE identifies and exploits such features using principles that are derived from the foundations of quantum mechanics, and are consequently agnostic any specific, high-level chemistry arguments. Identifying such chemical arguments would illustrate the role contextuality plays in chemistry, which would be of independent interest.

By using CS-VQE it is possible to reach chemical accuracy for ground state energies of numerous small molecules using many fewer qubits than would be required to implement full VQE on the tapered Hamitonians.
The number of terms and thus number of measurements required is also substantially reduced by using CS-VQE, since groups of terms become equivalent under the symmetry imposed by the noncontextual ground state. The number of measurements needed to obtain the quantum correction could be further reduced by the techniques described in~\cite{verteletskyi20a,yen20a,gokhale19a,izmaylov19a,zhao20a}. We leave this and other optimizations of the method to future work. Current VQE implementations are limited in both qubit count and number of measurements by the available hardware, so we expect CS-VQE to be of immediate practical value in accessing new molecular simulation applications on NISQ computers.

~
\begin{acknowledgements}
W. M. K. acknowledges support from the National Science Foundation, Grant No. DGE-1842474.
P. J. L. acknowledges support from the National Science Foundation, Grant No. PHY-1720395, and from Google Inc.
This work was supported by the NSF STAQ project (PHY-1818914).
\end{acknowledgements}

\bibliographystyle{apsrev4-1}
\bibliography{references.bib}

\begin{thebibliography}{55}%
\makeatletter
\providecommand \@ifxundefined [1]{%
 \@ifx{#1\undefined}
}%
\providecommand \@ifnum [1]{%
 \ifnum #1\expandafter \@firstoftwo
 \else \expandafter \@secondoftwo
 \fi
}%
\providecommand \@ifx [1]{%
 \ifx #1\expandafter \@firstoftwo
 \else \expandafter \@secondoftwo
 \fi
}%
\providecommand \natexlab [1]{#1}%
\providecommand \enquote  [1]{``#1''}%
\providecommand \bibnamefont  [1]{#1}%
\providecommand \bibfnamefont [1]{#1}%
\providecommand \citenamefont [1]{#1}%
\providecommand \href@noop [0]{\@secondoftwo}%
\providecommand \href [0]{\begingroup \@sanitize@url \@href}%
\providecommand \@href[1]{\@@startlink{#1}\@@href}%
\providecommand \@@href[1]{\endgroup#1\@@endlink}%
\providecommand \@sanitize@url [0]{\catcode `\\12\catcode `\$12\catcode
  `\&12\catcode `\#12\catcode `\^12\catcode `\_12\catcode `\%12\relax}%
\providecommand \@@startlink[1]{}%
\providecommand \@@endlink[0]{}%
\providecommand \url  [0]{\begingroup\@sanitize@url \@url }%
\providecommand \@url [1]{\endgroup\@href {#1}{\urlprefix }}%
\providecommand \urlprefix  [0]{URL }%
\providecommand \Eprint [0]{\href }%
\providecommand \doibase [0]{http://dx.doi.org/}%
\providecommand \selectlanguage [0]{\@gobble}%
\providecommand \bibinfo  [0]{\@secondoftwo}%
\providecommand \bibfield  [0]{\@secondoftwo}%
\providecommand \translation [1]{[#1]}%
\providecommand \BibitemOpen [0]{}%
\providecommand \bibitemStop [0]{}%
\providecommand \bibitemNoStop [0]{.\EOS\space}%
\providecommand \EOS [0]{\spacefactor3000\relax}%
\providecommand \BibitemShut  [1]{\csname bibitem#1\endcsname}%
\let\auto@bib@innerbib\@empty
\bibitem [{\citenamefont {Peruzzo}\ \emph {et~al.}(2014)\citenamefont
  {Peruzzo}, \citenamefont {McClean}, \citenamefont {Shadbolt}, \citenamefont
  {Yung}, \citenamefont {Zhou}, \citenamefont {Love}, \citenamefont
  {Aspuru-Guzik},\ and\ \citenamefont {O'Brien}}]{peruzzo14a}%
  \BibitemOpen
  \bibfield  {author} {\bibinfo {author} {\bibfnamefont {A.}~\bibnamefont
  {Peruzzo}}, \bibinfo {author} {\bibfnamefont {J.}~\bibnamefont {McClean}},
  \bibinfo {author} {\bibfnamefont {P.}~\bibnamefont {Shadbolt}}, \bibinfo
  {author} {\bibfnamefont {M.-H.}\ \bibnamefont {Yung}}, \bibinfo {author}
  {\bibfnamefont {X.-Q.}\ \bibnamefont {Zhou}}, \bibinfo {author}
  {\bibfnamefont {P.~J.}\ \bibnamefont {Love}}, \bibinfo {author}
  {\bibfnamefont {A.}~\bibnamefont {Aspuru-Guzik}}, \ and\ \bibinfo {author}
  {\bibfnamefont {J.~L.}\ \bibnamefont {O'Brien}},\ }\href
  {https://doi.org/10.1038/ncomms5213} {\bibfield  {journal} {\bibinfo
  {journal} {Nature Communications}\ }\textbf {\bibinfo {volume} {5}},\
  \bibinfo {pages} {4213 EP } (\bibinfo {year} {2014})}\BibitemShut {NoStop}%
\bibitem [{\citenamefont {O'Malley}\ \emph {et~al.}(2016)\citenamefont
  {O'Malley}, \citenamefont {Babbush}, \citenamefont {Kivlichan}, \citenamefont
  {Romero}, \citenamefont {McClean}, \citenamefont {Barends}, \citenamefont
  {Kelly}, \citenamefont {Roushan}, \citenamefont {Tranter}, \citenamefont
  {Ding}, \citenamefont {Campbell}, \citenamefont {Chen}, \citenamefont {Chen},
  \citenamefont {Chiaro}, \citenamefont {Dunsworth}, \citenamefont {Fowler},
  \citenamefont {Jeffrey}, \citenamefont {Lucero}, \citenamefont {Megrant},
  \citenamefont {Mutus}, \citenamefont {Neeley}, \citenamefont {Neill},
  \citenamefont {Quintana}, \citenamefont {Sank}, \citenamefont {Vainsencher},
  \citenamefont {Wenner}, \citenamefont {White}, \citenamefont {Coveney},
  \citenamefont {Love}, \citenamefont {Neven}, \citenamefont {Aspuru-Guzik},\
  and\ \citenamefont {Martinis}}]{omalley16a}%
  \BibitemOpen
  \bibfield  {author} {\bibinfo {author} {\bibfnamefont {P.~J.~J.}\
  \bibnamefont {O'Malley}}, \bibinfo {author} {\bibfnamefont {R.}~\bibnamefont
  {Babbush}}, \bibinfo {author} {\bibfnamefont {I.~D.}\ \bibnamefont
  {Kivlichan}}, \bibinfo {author} {\bibfnamefont {J.}~\bibnamefont {Romero}},
  \bibinfo {author} {\bibfnamefont {J.~R.}\ \bibnamefont {McClean}}, \bibinfo
  {author} {\bibfnamefont {R.}~\bibnamefont {Barends}}, \bibinfo {author}
  {\bibfnamefont {J.}~\bibnamefont {Kelly}}, \bibinfo {author} {\bibfnamefont
  {P.}~\bibnamefont {Roushan}}, \bibinfo {author} {\bibfnamefont
  {A.}~\bibnamefont {Tranter}}, \bibinfo {author} {\bibfnamefont
  {N.}~\bibnamefont {Ding}}, \bibinfo {author} {\bibfnamefont {B.}~\bibnamefont
  {Campbell}}, \bibinfo {author} {\bibfnamefont {Y.}~\bibnamefont {Chen}},
  \bibinfo {author} {\bibfnamefont {Z.}~\bibnamefont {Chen}}, \bibinfo {author}
  {\bibfnamefont {B.}~\bibnamefont {Chiaro}}, \bibinfo {author} {\bibfnamefont
  {A.}~\bibnamefont {Dunsworth}}, \bibinfo {author} {\bibfnamefont {A.~G.}\
  \bibnamefont {Fowler}}, \bibinfo {author} {\bibfnamefont {E.}~\bibnamefont
  {Jeffrey}}, \bibinfo {author} {\bibfnamefont {E.}~\bibnamefont {Lucero}},
  \bibinfo {author} {\bibfnamefont {A.}~\bibnamefont {Megrant}}, \bibinfo
  {author} {\bibfnamefont {J.~Y.}\ \bibnamefont {Mutus}}, \bibinfo {author}
  {\bibfnamefont {M.}~\bibnamefont {Neeley}}, \bibinfo {author} {\bibfnamefont
  {C.}~\bibnamefont {Neill}}, \bibinfo {author} {\bibfnamefont
  {C.}~\bibnamefont {Quintana}}, \bibinfo {author} {\bibfnamefont
  {D.}~\bibnamefont {Sank}}, \bibinfo {author} {\bibfnamefont {A.}~\bibnamefont
  {Vainsencher}}, \bibinfo {author} {\bibfnamefont {J.}~\bibnamefont {Wenner}},
  \bibinfo {author} {\bibfnamefont {T.~C.}\ \bibnamefont {White}}, \bibinfo
  {author} {\bibfnamefont {P.~V.}\ \bibnamefont {Coveney}}, \bibinfo {author}
  {\bibfnamefont {P.~J.}\ \bibnamefont {Love}}, \bibinfo {author}
  {\bibfnamefont {H.}~\bibnamefont {Neven}}, \bibinfo {author} {\bibfnamefont
  {A.}~\bibnamefont {Aspuru-Guzik}}, \ and\ \bibinfo {author} {\bibfnamefont
  {J.~M.}\ \bibnamefont {Martinis}},\ }\href
  {https://doi.org/10.1103/PhysRevX.6.031007} {\bibfield  {journal} {\bibinfo
  {journal} {Phys. Rev. X}\ }\textbf {\bibinfo {volume} {6}},\ \bibinfo {pages}
  {031007} (\bibinfo {year} {2016})}\BibitemShut {NoStop}%
\bibitem [{\citenamefont {Santagati}\ \emph {et~al.}(2018)\citenamefont
  {Santagati}, \citenamefont {Wang}, \citenamefont {Gentile}, \citenamefont
  {Paesani}, \citenamefont {Wiebe}, \citenamefont {McClean}, \citenamefont
  {Morley-Short}, \citenamefont {Shadbolt}, \citenamefont {Bonneau},
  \citenamefont {Silverstone}, \citenamefont {Tew}, \citenamefont {Zhou},
  \citenamefont {O{\textquoteright}Brien},\ and\ \citenamefont
  {Thompson}}]{santagati18a}%
  \BibitemOpen
  \bibfield  {author} {\bibinfo {author} {\bibfnamefont {R.}~\bibnamefont
  {Santagati}}, \bibinfo {author} {\bibfnamefont {J.}~\bibnamefont {Wang}},
  \bibinfo {author} {\bibfnamefont {A.~A.}\ \bibnamefont {Gentile}}, \bibinfo
  {author} {\bibfnamefont {S.}~\bibnamefont {Paesani}}, \bibinfo {author}
  {\bibfnamefont {N.}~\bibnamefont {Wiebe}}, \bibinfo {author} {\bibfnamefont
  {J.~R.}\ \bibnamefont {McClean}}, \bibinfo {author} {\bibfnamefont
  {S.}~\bibnamefont {Morley-Short}}, \bibinfo {author} {\bibfnamefont {P.~J.}\
  \bibnamefont {Shadbolt}}, \bibinfo {author} {\bibfnamefont {D.}~\bibnamefont
  {Bonneau}}, \bibinfo {author} {\bibfnamefont {J.~W.}\ \bibnamefont
  {Silverstone}}, \bibinfo {author} {\bibfnamefont {D.~P.}\ \bibnamefont
  {Tew}}, \bibinfo {author} {\bibfnamefont {X.}~\bibnamefont {Zhou}}, \bibinfo
  {author} {\bibfnamefont {J.~L.}\ \bibnamefont {O{\textquoteright}Brien}}, \
  and\ \bibinfo {author} {\bibfnamefont {M.~G.}\ \bibnamefont {Thompson}},\
  }\href {https://doi.org/10.1126/sciadv.aap9646} {\bibfield  {journal}
  {\bibinfo  {journal} {Science Advances}\ }\textbf {\bibinfo {volume} {4}}
  (\bibinfo {year} {2018})}\BibitemShut {NoStop}%
\bibitem [{\citenamefont {Shen}\ \emph {et~al.}(2017)\citenamefont {Shen},
  \citenamefont {Zhang}, \citenamefont {Zhang}, \citenamefont {Zhang},
  \citenamefont {Yung},\ and\ \citenamefont {Kim}}]{shen2017quantum}%
  \BibitemOpen
  \bibfield  {author} {\bibinfo {author} {\bibfnamefont {Y.}~\bibnamefont
  {Shen}}, \bibinfo {author} {\bibfnamefont {X.}~\bibnamefont {Zhang}},
  \bibinfo {author} {\bibfnamefont {S.}~\bibnamefont {Zhang}}, \bibinfo
  {author} {\bibfnamefont {J.-N.}\ \bibnamefont {Zhang}}, \bibinfo {author}
  {\bibfnamefont {M.-H.}\ \bibnamefont {Yung}}, \ and\ \bibinfo {author}
  {\bibfnamefont {K.}~\bibnamefont {Kim}},\ }\href
  {https://doi.org/10.1103/PhysRevA.95.020501} {\bibfield  {journal} {\bibinfo
  {journal} {Physical Review A}\ }\textbf {\bibinfo {volume} {95}},\ \bibinfo
  {pages} {020501} (\bibinfo {year} {2017})}\BibitemShut {NoStop}%
\bibitem [{\citenamefont {Paesani}\ \emph {et~al.}(2017)\citenamefont
  {Paesani}, \citenamefont {Gentile}, \citenamefont {Santagati}, \citenamefont
  {Wang}, \citenamefont {Wiebe}, \citenamefont {Tew}, \citenamefont
  {O'{B}rien},\ and\ \citenamefont {Thompson}}]{paesani2017}%
  \BibitemOpen
  \bibfield  {author} {\bibinfo {author} {\bibfnamefont {S.}~\bibnamefont
  {Paesani}}, \bibinfo {author} {\bibfnamefont {A.~A.}\ \bibnamefont
  {Gentile}}, \bibinfo {author} {\bibfnamefont {R.}~\bibnamefont {Santagati}},
  \bibinfo {author} {\bibfnamefont {J.}~\bibnamefont {Wang}}, \bibinfo {author}
  {\bibfnamefont {N.}~\bibnamefont {Wiebe}}, \bibinfo {author} {\bibfnamefont
  {D.~P.}\ \bibnamefont {Tew}}, \bibinfo {author} {\bibfnamefont {J.~L.}\
  \bibnamefont {O'{B}rien}}, \ and\ \bibinfo {author} {\bibfnamefont {M.~G.}\
  \bibnamefont {Thompson}},\ }\href
  {https://doi.org/10.1103/PhysRevLett.118.100503} {\bibfield  {journal}
  {\bibinfo  {journal} {Phys. Rev. Lett.}\ }\textbf {\bibinfo {volume} {118}},\
  \bibinfo {pages} {100503} (\bibinfo {year} {2017})}\BibitemShut {NoStop}%
\bibitem [{\citenamefont {Kandala}\ \emph {et~al.}(2017)\citenamefont
  {Kandala}, \citenamefont {Mezzacapo}, \citenamefont {Temme}, \citenamefont
  {Takita}, \citenamefont {Brink}, \citenamefont {Chow},\ and\ \citenamefont
  {Gambetta}}]{kandala17a}%
  \BibitemOpen
  \bibfield  {author} {\bibinfo {author} {\bibfnamefont {A.}~\bibnamefont
  {Kandala}}, \bibinfo {author} {\bibfnamefont {A.}~\bibnamefont {Mezzacapo}},
  \bibinfo {author} {\bibfnamefont {K.}~\bibnamefont {Temme}}, \bibinfo
  {author} {\bibfnamefont {M.}~\bibnamefont {Takita}}, \bibinfo {author}
  {\bibfnamefont {M.}~\bibnamefont {Brink}}, \bibinfo {author} {\bibfnamefont
  {J.~M.}\ \bibnamefont {Chow}}, \ and\ \bibinfo {author} {\bibfnamefont
  {J.~M.}\ \bibnamefont {Gambetta}},\ }\href
  {https://doi.org/10.1038/nature23879} {\bibfield  {journal} {\bibinfo
  {journal} {Nature}\ }\textbf {\bibinfo {volume} {549}},\ \bibinfo {pages}
  {242} (\bibinfo {year} {2017})}\BibitemShut {NoStop}%
\bibitem [{\citenamefont {Hempel}\ \emph {et~al.}(2018)\citenamefont {Hempel},
  \citenamefont {Maier}, \citenamefont {Romero}, \citenamefont {McClean},
  \citenamefont {Monz}, \citenamefont {Shen}, \citenamefont {Jurcevic},
  \citenamefont {Lanyon}, \citenamefont {Love}, \citenamefont {Babbush},
  \citenamefont {Aspuru-Guzik}, \citenamefont {Blatt},\ and\ \citenamefont
  {Roos}}]{hempel18a}%
  \BibitemOpen
  \bibfield  {author} {\bibinfo {author} {\bibfnamefont {C.}~\bibnamefont
  {Hempel}}, \bibinfo {author} {\bibfnamefont {C.}~\bibnamefont {Maier}},
  \bibinfo {author} {\bibfnamefont {J.}~\bibnamefont {Romero}}, \bibinfo
  {author} {\bibfnamefont {J.}~\bibnamefont {McClean}}, \bibinfo {author}
  {\bibfnamefont {T.}~\bibnamefont {Monz}}, \bibinfo {author} {\bibfnamefont
  {H.}~\bibnamefont {Shen}}, \bibinfo {author} {\bibfnamefont {P.}~\bibnamefont
  {Jurcevic}}, \bibinfo {author} {\bibfnamefont {B.~P.}\ \bibnamefont
  {Lanyon}}, \bibinfo {author} {\bibfnamefont {P.}~\bibnamefont {Love}},
  \bibinfo {author} {\bibfnamefont {R.}~\bibnamefont {Babbush}}, \bibinfo
  {author} {\bibfnamefont {A.}~\bibnamefont {Aspuru-Guzik}}, \bibinfo {author}
  {\bibfnamefont {R.}~\bibnamefont {Blatt}}, \ and\ \bibinfo {author}
  {\bibfnamefont {C.~F.}\ \bibnamefont {Roos}},\ }\href
  {https://doi.org/10.1103/PhysRevX.8.031022} {\bibfield  {journal} {\bibinfo
  {journal} {Phys. Rev. X}\ }\textbf {\bibinfo {volume} {8}},\ \bibinfo {pages}
  {031022} (\bibinfo {year} {2018})}\BibitemShut {NoStop}%
\bibitem [{\citenamefont {Dumitrescu}\ \emph {et~al.}(2018)\citenamefont
  {Dumitrescu}, \citenamefont {McCaskey}, \citenamefont {Hagen}, \citenamefont
  {Jansen}, \citenamefont {Morris}, \citenamefont {Papenbrock}, \citenamefont
  {Pooser}, \citenamefont {Dean},\ and\ \citenamefont
  {Lougovski}}]{dumitrescu18a}%
  \BibitemOpen
  \bibfield  {author} {\bibinfo {author} {\bibfnamefont {E.~F.}\ \bibnamefont
  {Dumitrescu}}, \bibinfo {author} {\bibfnamefont {A.~J.}\ \bibnamefont
  {McCaskey}}, \bibinfo {author} {\bibfnamefont {G.}~\bibnamefont {Hagen}},
  \bibinfo {author} {\bibfnamefont {G.~R.}\ \bibnamefont {Jansen}}, \bibinfo
  {author} {\bibfnamefont {T.~D.}\ \bibnamefont {Morris}}, \bibinfo {author}
  {\bibfnamefont {T.}~\bibnamefont {Papenbrock}}, \bibinfo {author}
  {\bibfnamefont {R.~C.}\ \bibnamefont {Pooser}}, \bibinfo {author}
  {\bibfnamefont {D.~J.}\ \bibnamefont {Dean}}, \ and\ \bibinfo {author}
  {\bibfnamefont {P.}~\bibnamefont {Lougovski}},\ }\href
  {https://doi.org/10.1103/PhysRevLett.120.210501} {\bibfield  {journal}
  {\bibinfo  {journal} {Phys. Rev. Lett.}\ }\textbf {\bibinfo {volume} {120}},\
  \bibinfo {pages} {210501} (\bibinfo {year} {2018})}\BibitemShut {NoStop}%
\bibitem [{\citenamefont {Colless}\ \emph {et~al.}(2018)\citenamefont
  {Colless}, \citenamefont {Ramasesh}, \citenamefont {Dahlen}, \citenamefont
  {Blok}, \citenamefont {Kimchi-Schwartz}, \citenamefont {McClean},
  \citenamefont {Carter}, \citenamefont {de~Jong},\ and\ \citenamefont
  {Siddiqi}}]{colless18a}%
  \BibitemOpen
  \bibfield  {author} {\bibinfo {author} {\bibfnamefont {J.~I.}\ \bibnamefont
  {Colless}}, \bibinfo {author} {\bibfnamefont {V.~V.}\ \bibnamefont
  {Ramasesh}}, \bibinfo {author} {\bibfnamefont {D.}~\bibnamefont {Dahlen}},
  \bibinfo {author} {\bibfnamefont {M.~S.}\ \bibnamefont {Blok}}, \bibinfo
  {author} {\bibfnamefont {M.~E.}\ \bibnamefont {Kimchi-Schwartz}}, \bibinfo
  {author} {\bibfnamefont {J.~R.}\ \bibnamefont {McClean}}, \bibinfo {author}
  {\bibfnamefont {J.}~\bibnamefont {Carter}}, \bibinfo {author} {\bibfnamefont
  {W.~A.}\ \bibnamefont {de~Jong}}, \ and\ \bibinfo {author} {\bibfnamefont
  {I.}~\bibnamefont {Siddiqi}},\ }\href {\doibase 10.1103/PhysRevX.8.011021}
  {\bibfield  {journal} {\bibinfo  {journal} {Phys. Rev. X}\ }\textbf {\bibinfo
  {volume} {8}},\ \bibinfo {pages} {011021} (\bibinfo {year}
  {2018})}\BibitemShut {NoStop}%
\bibitem [{\citenamefont {Nam}\ \emph {et~al.}(2020)\citenamefont {Nam},
  \citenamefont {Chen}, \citenamefont {Pisenti}, \citenamefont {Wright},
  \citenamefont {Delaney}, \citenamefont {Maslov}, \citenamefont {Brown},
  \citenamefont {Allen}, \citenamefont {Amini}, \citenamefont {Apisdorf},
  \citenamefont {Beck}, \citenamefont {Blinov}, \citenamefont {Chaplin},
  \citenamefont {Chmielewski}, \citenamefont {Collins}, \citenamefont
  {Debnath}, \citenamefont {Hudek}, \citenamefont {Ducore}, \citenamefont
  {Keesan}, \citenamefont {Kreikemeier}, \citenamefont {Mizrahi}, \citenamefont
  {Solomon}, \citenamefont {Williams}, \citenamefont {Wong-Campos},
  \citenamefont {Moehring}, \citenamefont {Monroe},\ and\ \citenamefont
  {Kim}}]{nam19a}%
  \BibitemOpen
  \bibfield  {author} {\bibinfo {author} {\bibfnamefont {Y.}~\bibnamefont
  {Nam}}, \bibinfo {author} {\bibfnamefont {J.-S.}\ \bibnamefont {Chen}},
  \bibinfo {author} {\bibfnamefont {N.~C.}\ \bibnamefont {Pisenti}}, \bibinfo
  {author} {\bibfnamefont {K.}~\bibnamefont {Wright}}, \bibinfo {author}
  {\bibfnamefont {C.}~\bibnamefont {Delaney}}, \bibinfo {author} {\bibfnamefont
  {D.}~\bibnamefont {Maslov}}, \bibinfo {author} {\bibfnamefont {K.~R.}\
  \bibnamefont {Brown}}, \bibinfo {author} {\bibfnamefont {S.}~\bibnamefont
  {Allen}}, \bibinfo {author} {\bibfnamefont {J.~M.}\ \bibnamefont {Amini}},
  \bibinfo {author} {\bibfnamefont {J.}~\bibnamefont {Apisdorf}}, \bibinfo
  {author} {\bibfnamefont {K.~M.}\ \bibnamefont {Beck}}, \bibinfo {author}
  {\bibfnamefont {A.}~\bibnamefont {Blinov}}, \bibinfo {author} {\bibfnamefont
  {V.}~\bibnamefont {Chaplin}}, \bibinfo {author} {\bibfnamefont
  {M.}~\bibnamefont {Chmielewski}}, \bibinfo {author} {\bibfnamefont
  {C.}~\bibnamefont {Collins}}, \bibinfo {author} {\bibfnamefont
  {S.}~\bibnamefont {Debnath}}, \bibinfo {author} {\bibfnamefont {K.~M.}\
  \bibnamefont {Hudek}}, \bibinfo {author} {\bibfnamefont {A.~M.}\ \bibnamefont
  {Ducore}}, \bibinfo {author} {\bibfnamefont {M.}~\bibnamefont {Keesan}},
  \bibinfo {author} {\bibfnamefont {S.~M.}\ \bibnamefont {Kreikemeier}},
  \bibinfo {author} {\bibfnamefont {J.}~\bibnamefont {Mizrahi}}, \bibinfo
  {author} {\bibfnamefont {P.}~\bibnamefont {Solomon}}, \bibinfo {author}
  {\bibfnamefont {M.}~\bibnamefont {Williams}}, \bibinfo {author}
  {\bibfnamefont {J.~D.}\ \bibnamefont {Wong-Campos}}, \bibinfo {author}
  {\bibfnamefont {D.}~\bibnamefont {Moehring}}, \bibinfo {author}
  {\bibfnamefont {C.}~\bibnamefont {Monroe}}, \ and\ \bibinfo {author}
  {\bibfnamefont {J.}~\bibnamefont {Kim}},\ }\href
  {https://doi.org/10.1038/s41534-020-0259-3} {\bibfield  {journal} {\bibinfo
  {journal} {npj Quantum Information}\ }\textbf {\bibinfo {volume} {6}},\
  \bibinfo {pages} {33} (\bibinfo {year} {2020})}\BibitemShut {NoStop}%
\bibitem [{\citenamefont {Kokail}\ \emph {et~al.}(2019)\citenamefont {Kokail},
  \citenamefont {Maier}, \citenamefont {van Bijnen}, \citenamefont {Brydges},
  \citenamefont {Joshi}, \citenamefont {Jurcevic}, \citenamefont {Muschik},
  \citenamefont {Silvi}, \citenamefont {Blatt}, \citenamefont {Roos},\ and\
  \citenamefont {Zoller}}]{kokail19a}%
  \BibitemOpen
  \bibfield  {author} {\bibinfo {author} {\bibfnamefont {C.}~\bibnamefont
  {Kokail}}, \bibinfo {author} {\bibfnamefont {C.}~\bibnamefont {Maier}},
  \bibinfo {author} {\bibfnamefont {R.}~\bibnamefont {van Bijnen}}, \bibinfo
  {author} {\bibfnamefont {T.}~\bibnamefont {Brydges}}, \bibinfo {author}
  {\bibfnamefont {M.~K.}\ \bibnamefont {Joshi}}, \bibinfo {author}
  {\bibfnamefont {P.}~\bibnamefont {Jurcevic}}, \bibinfo {author}
  {\bibfnamefont {C.~A.}\ \bibnamefont {Muschik}}, \bibinfo {author}
  {\bibfnamefont {P.}~\bibnamefont {Silvi}}, \bibinfo {author} {\bibfnamefont
  {R.}~\bibnamefont {Blatt}}, \bibinfo {author} {\bibfnamefont {C.~F.}\
  \bibnamefont {Roos}}, \ and\ \bibinfo {author} {\bibfnamefont
  {P.}~\bibnamefont {Zoller}},\ }\href {\doibase 10.1038/s41586-019-1177-4}
  {\bibfield  {journal} {\bibinfo  {journal} {Nature}\ }\textbf {\bibinfo
  {volume} {569}},\ \bibinfo {pages} {355} (\bibinfo {year}
  {2019})}\BibitemShut {NoStop}%
\bibitem [{\citenamefont {Kandala}\ \emph {et~al.}(2019)\citenamefont
  {Kandala}, \citenamefont {Temme}, \citenamefont {C{\'o}rcoles}, \citenamefont
  {Mezzacapo}, \citenamefont {Chow},\ and\ \citenamefont
  {Gambetta}}]{kandala19a}%
  \BibitemOpen
  \bibfield  {author} {\bibinfo {author} {\bibfnamefont {A.}~\bibnamefont
  {Kandala}}, \bibinfo {author} {\bibfnamefont {K.}~\bibnamefont {Temme}},
  \bibinfo {author} {\bibfnamefont {A.~D.}\ \bibnamefont {C{\'o}rcoles}},
  \bibinfo {author} {\bibfnamefont {A.}~\bibnamefont {Mezzacapo}}, \bibinfo
  {author} {\bibfnamefont {J.~M.}\ \bibnamefont {Chow}}, \ and\ \bibinfo
  {author} {\bibfnamefont {J.~M.}\ \bibnamefont {Gambetta}},\ }\href
  {https://doi.org/10.1038/s41586-019-1040-7} {\bibfield  {journal} {\bibinfo
  {journal} {Nature}\ }\textbf {\bibinfo {volume} {567}},\ \bibinfo {pages}
  {491} (\bibinfo {year} {2019})}\BibitemShut {NoStop}%
\bibitem [{\citenamefont {{Google AI Quantum and
  Collaborators}}(2020)}]{google20a}%
  \BibitemOpen
  \bibfield  {author} {\bibinfo {author} {\bibnamefont {{Google AI Quantum and
  Collaborators}}},\ }\href {\doibase 10.1126/science.abb9811} {\bibfield
  {journal} {\bibinfo  {journal} {Science}\ }\textbf {\bibinfo {volume}
  {369}},\ \bibinfo {pages} {1084} (\bibinfo {year} {2020})}\BibitemShut
  {NoStop}%
\bibitem [{\citenamefont {Kirby}\ and\ \citenamefont {Love}(2019)}]{kirby19a}%
  \BibitemOpen
  \bibfield  {author} {\bibinfo {author} {\bibfnamefont {W.~M.}\ \bibnamefont
  {Kirby}}\ and\ \bibinfo {author} {\bibfnamefont {P.~J.}\ \bibnamefont
  {Love}},\ }\href {\doibase 10.1103/PhysRevLett.123.200501} {\bibfield
  {journal} {\bibinfo  {journal} {Phys. Rev. Lett.}\ }\textbf {\bibinfo
  {volume} {123}},\ \bibinfo {pages} {200501} (\bibinfo {year}
  {2019})}\BibitemShut {NoStop}%
\bibitem [{\citenamefont {Kirby}\ and\ \citenamefont {Love}(2020)}]{kirby20a}%
  \BibitemOpen
  \bibfield  {author} {\bibinfo {author} {\bibfnamefont {W.~M.}\ \bibnamefont
  {Kirby}}\ and\ \bibinfo {author} {\bibfnamefont {P.~J.}\ \bibnamefont
  {Love}},\ }\href {\doibase 10.1103/PhysRevA.102.032418} {\bibfield  {journal}
  {\bibinfo  {journal} {Phys. Rev. A}\ }\textbf {\bibinfo {volume} {102}},\
  \bibinfo {pages} {032418} (\bibinfo {year} {2020})}\BibitemShut {NoStop}%
\bibitem [{\citenamefont {Spekkens}(2007)}]{spekkens07a}%
  \BibitemOpen
  \bibfield  {author} {\bibinfo {author} {\bibfnamefont {R.~W.}\ \bibnamefont
  {Spekkens}},\ }\href {https://doi.org/10.1103/PhysRevA.75.032110} {\bibfield
  {journal} {\bibinfo  {journal} {Phys. Rev. A}\ }\textbf {\bibinfo {volume}
  {75}},\ \bibinfo {pages} {032110} (\bibinfo {year} {2007})}\BibitemShut
  {NoStop}%
\bibitem [{\citenamefont {Spekkens}(2016)}]{spekkens16a}%
  \BibitemOpen
  \bibfield  {author} {\bibinfo {author} {\bibfnamefont {R.~W.}\ \bibnamefont
  {Spekkens}},\ }\enquote {\bibinfo {title} {Quasi-quantization: Classical
  statistical theories with an epistemic restriction},}\ in\ \href
  {https://doi.org/10.1007/978-94-017-7303-4_4} {\emph {\bibinfo {booktitle}
  {Quantum Theory: Informational Foundations and Foils}}},\ \bibinfo {editor}
  {edited by\ \bibinfo {editor} {\bibfnamefont {G.}~\bibnamefont {Chiribella}}\
  and\ \bibinfo {editor} {\bibfnamefont {R.~W.}\ \bibnamefont {Spekkens}}}\
  (\bibinfo  {publisher} {Springer Netherlands},\ \bibinfo {address}
  {Dordrecht},\ \bibinfo {year} {2016})\ pp.\ \bibinfo {pages}
  {83--135}\BibitemShut {NoStop}%
\bibitem [{\citenamefont {Nakanishi}\ \emph {et~al.}(2019)\citenamefont
  {Nakanishi}, \citenamefont {Mitarai},\ and\ \citenamefont
  {Fujii}}]{nakanishi19a}%
  \BibitemOpen
  \bibfield  {author} {\bibinfo {author} {\bibfnamefont {K.~M.}\ \bibnamefont
  {Nakanishi}}, \bibinfo {author} {\bibfnamefont {K.}~\bibnamefont {Mitarai}},
  \ and\ \bibinfo {author} {\bibfnamefont {K.}~\bibnamefont {Fujii}},\ }\href
  {https://doi.org/10.1103/PhysRevResearch.1.033062} {\bibfield  {journal}
  {\bibinfo  {journal} {Phys. Rev. Research}\ }\textbf {\bibinfo {volume}
  {1}},\ \bibinfo {pages} {033062} (\bibinfo {year} {2019})}\BibitemShut
  {NoStop}%
\bibitem [{\citenamefont {Bell}(1964)}]{bell64a}%
  \BibitemOpen
  \bibfield  {author} {\bibinfo {author} {\bibfnamefont {J.~S.}\ \bibnamefont
  {Bell}},\ }\href {https://doi.org/10.1103/PhysicsPhysiqueFizika.1.195}
  {\bibfield  {journal} {\bibinfo  {journal} {Physics}\ }\textbf {\bibinfo
  {volume} {1}},\ \bibinfo {pages} {195} (\bibinfo {year} {1964})}\BibitemShut
  {NoStop}%
\bibitem [{\citenamefont {Bell}(1966)}]{bell66a}%
  \BibitemOpen
  \bibfield  {author} {\bibinfo {author} {\bibfnamefont {J.~S.}\ \bibnamefont
  {Bell}},\ }\href {https://doi.org/10.1103/RevModPhys.38.447} {\bibfield
  {journal} {\bibinfo  {journal} {Rev. Mod. Phys.}\ }\textbf {\bibinfo {volume}
  {38}},\ \bibinfo {pages} {447} (\bibinfo {year} {1966})}\BibitemShut
  {NoStop}%
\bibitem [{\citenamefont {Kochen}\ and\ \citenamefont
  {Specker}(1967)}]{kochen67a}%
  \BibitemOpen
  \bibfield  {author} {\bibinfo {author} {\bibfnamefont {S.}~\bibnamefont
  {Kochen}}\ and\ \bibinfo {author} {\bibfnamefont {E.}~\bibnamefont
  {Specker}},\ }\href {https://doi.org/10.1007/978-94-010-1795-4_17} {\bibfield
   {journal} {\bibinfo  {journal} {J. Math. Mech.}\ }\textbf {\bibinfo {volume}
  {17}},\ \bibinfo {pages} {59} (\bibinfo {year} {1967})}\BibitemShut {NoStop}%
\bibitem [{\citenamefont {Spekkens}(2005)}]{spekkens05a}%
  \BibitemOpen
  \bibfield  {author} {\bibinfo {author} {\bibfnamefont {R.~W.}\ \bibnamefont
  {Spekkens}},\ }\href {https://doi.org/10.1103/PhysRevA.71.052108} {\bibfield
  {journal} {\bibinfo  {journal} {Phys. Rev. A}\ }\textbf {\bibinfo {volume}
  {71}},\ \bibinfo {pages} {052108} (\bibinfo {year} {2005})}\BibitemShut
  {NoStop}%
\bibitem [{\citenamefont {Abramsky}\ and\ \citenamefont
  {Brandenburger}(2011)}]{abramsky11a}%
  \BibitemOpen
  \bibfield  {author} {\bibinfo {author} {\bibfnamefont {S.}~\bibnamefont
  {Abramsky}}\ and\ \bibinfo {author} {\bibfnamefont {A.}~\bibnamefont
  {Brandenburger}},\ }\href {https://doi.org/10.1088/1367-2630/13/11/113036}
  {\bibfield  {journal} {\bibinfo  {journal} {New Journal of Physics}\ }\textbf
  {\bibinfo {volume} {13}},\ \bibinfo {pages} {113036} (\bibinfo {year}
  {2011})}\BibitemShut {NoStop}%
\bibitem [{\citenamefont {Raussendorf}(2013)}]{raussendorf13a}%
  \BibitemOpen
  \bibfield  {author} {\bibinfo {author} {\bibfnamefont {R.}~\bibnamefont
  {Raussendorf}},\ }\href {https://doi.org/10.1103/PhysRevA.88.022322}
  {\bibfield  {journal} {\bibinfo  {journal} {Phys. Rev. A}\ }\textbf {\bibinfo
  {volume} {88}},\ \bibinfo {pages} {022322} (\bibinfo {year}
  {2013})}\BibitemShut {NoStop}%
\bibitem [{\citenamefont {Howard}\ \emph {et~al.}(2014)\citenamefont {Howard},
  \citenamefont {Wallman}, \citenamefont {Veitch},\ and\ \citenamefont
  {Emerson}}]{howard14a}%
  \BibitemOpen
  \bibfield  {author} {\bibinfo {author} {\bibfnamefont {M.}~\bibnamefont
  {Howard}}, \bibinfo {author} {\bibfnamefont {J.}~\bibnamefont {Wallman}},
  \bibinfo {author} {\bibfnamefont {V.}~\bibnamefont {Veitch}}, \ and\ \bibinfo
  {author} {\bibfnamefont {J.}~\bibnamefont {Emerson}},\ }\href
  {https://doi.org/10.1038/nature13460} {\bibfield  {journal} {\bibinfo
  {journal} {Nature}\ }\textbf {\bibinfo {volume} {510}},\ \bibinfo {pages}
  {351 EP } (\bibinfo {year} {2014})}\BibitemShut {NoStop}%
\bibitem [{\citenamefont {Cabello}\ \emph {et~al.}(2014)\citenamefont
  {Cabello}, \citenamefont {Severini},\ and\ \citenamefont
  {Winter}}]{cabello14a}%
  \BibitemOpen
  \bibfield  {author} {\bibinfo {author} {\bibfnamefont {A.}~\bibnamefont
  {Cabello}}, \bibinfo {author} {\bibfnamefont {S.}~\bibnamefont {Severini}}, \
  and\ \bibinfo {author} {\bibfnamefont {A.}~\bibnamefont {Winter}},\ }\href
  {https://doi.org/10.1103/PhysRevLett.112.040401} {\bibfield  {journal}
  {\bibinfo  {journal} {Phys. Rev. Lett.}\ }\textbf {\bibinfo {volume} {112}},\
  \bibinfo {pages} {040401} (\bibinfo {year} {2014})}\BibitemShut {NoStop}%
\bibitem [{\citenamefont {Cabello}\ \emph {et~al.}(2015)\citenamefont
  {Cabello}, \citenamefont {Kleinmann},\ and\ \citenamefont
  {Budroni}}]{cabello15a}%
  \BibitemOpen
  \bibfield  {author} {\bibinfo {author} {\bibfnamefont {A.}~\bibnamefont
  {Cabello}}, \bibinfo {author} {\bibfnamefont {M.}~\bibnamefont {Kleinmann}},
  \ and\ \bibinfo {author} {\bibfnamefont {C.}~\bibnamefont {Budroni}},\ }\href
  {https://doi.org/10.1103/PhysRevLett.114.250402} {\bibfield  {journal}
  {\bibinfo  {journal} {Phys. Rev. Lett.}\ }\textbf {\bibinfo {volume} {114}},\
  \bibinfo {pages} {250402} (\bibinfo {year} {2015})}\BibitemShut {NoStop}%
\bibitem [{\citenamefont {Ramanathan}\ and\ \citenamefont
  {Horodecki}(2014)}]{ramanathan14a}%
  \BibitemOpen
  \bibfield  {author} {\bibinfo {author} {\bibfnamefont {R.}~\bibnamefont
  {Ramanathan}}\ and\ \bibinfo {author} {\bibfnamefont {P.}~\bibnamefont
  {Horodecki}},\ }\href {https://doi.org/10.1103/PhysRevLett.112.040404}
  {\bibfield  {journal} {\bibinfo  {journal} {Phys. Rev. Lett.}\ }\textbf
  {\bibinfo {volume} {112}},\ \bibinfo {pages} {040404} (\bibinfo {year}
  {2014})}\BibitemShut {NoStop}%
\bibitem [{\citenamefont {de~Silva}(2017)}]{de_silva17a}%
  \BibitemOpen
  \bibfield  {author} {\bibinfo {author} {\bibfnamefont {N.}~\bibnamefont
  {de~Silva}},\ }\href {https://doi.org/10.1103/PhysRevA.95.032108} {\bibfield
  {journal} {\bibinfo  {journal} {Phys. Rev. A}\ }\textbf {\bibinfo {volume}
  {95}},\ \bibinfo {pages} {032108} (\bibinfo {year} {2017})}\BibitemShut
  {NoStop}%
\bibitem [{\citenamefont {Amaral}\ and\ \citenamefont
  {Cunha}(2017)}]{amaral17a}%
  \BibitemOpen
  \bibfield  {author} {\bibinfo {author} {\bibfnamefont {B.}~\bibnamefont
  {Amaral}}\ and\ \bibinfo {author} {\bibfnamefont {M.~T.}\ \bibnamefont
  {Cunha}},\ }\href@noop {} {\bibfield  {journal} {\bibinfo  {journal} {arXiv
  preprint}\ } (\bibinfo {year} {2017})},\ \Eprint
  {http://arxiv.org/abs/1709.04812} {arXiv:1709.04812 [quant-ph]} \BibitemShut
  {NoStop}%
\bibitem [{\citenamefont {Xu}\ and\ \citenamefont {Cabello}(2019)}]{xu18a}%
  \BibitemOpen
  \bibfield  {author} {\bibinfo {author} {\bibfnamefont {Z.-P.}\ \bibnamefont
  {Xu}}\ and\ \bibinfo {author} {\bibfnamefont {A.}~\bibnamefont {Cabello}},\
  }\href {https://doi.org/10.1103/PhysRevA.99.020103} {\bibfield  {journal}
  {\bibinfo  {journal} {Phys. Rev. A}\ }\textbf {\bibinfo {volume} {99}},\
  \bibinfo {pages} {020103} (\bibinfo {year} {2019})}\BibitemShut {NoStop}%
\bibitem [{\citenamefont {Raussendorf}(2019)}]{raussendorf19c}%
  \BibitemOpen
  \bibfield  {author} {\bibinfo {author} {\bibfnamefont {R.}~\bibnamefont
  {Raussendorf}},\ }\href {https://doi.org/10.26421/QIC19.13-14-4} {\bibfield
  {journal} {\bibinfo  {journal} {Quantum Information and Computation}\
  }\textbf {\bibinfo {volume} {19}},\ \bibinfo {pages} {1141} (\bibinfo {year}
  {2019})}\BibitemShut {NoStop}%
\bibitem [{\citenamefont {Duarte}\ and\ \citenamefont
  {Amaral}(2018)}]{duarte18a}%
  \BibitemOpen
  \bibfield  {author} {\bibinfo {author} {\bibfnamefont {C.}~\bibnamefont
  {Duarte}}\ and\ \bibinfo {author} {\bibfnamefont {B.}~\bibnamefont
  {Amaral}},\ }\href {\doibase 10.1063/1.5018582} {\bibfield  {journal}
  {\bibinfo  {journal} {Journal of Mathematical Physics}\ }\textbf {\bibinfo
  {volume} {59}},\ \bibinfo {pages} {062202} (\bibinfo {year}
  {2018})}\BibitemShut {NoStop}%
\bibitem [{\citenamefont {Mansfield}\ and\ \citenamefont
  {Kashefi}(2018)}]{mansfield18a}%
  \BibitemOpen
  \bibfield  {author} {\bibinfo {author} {\bibfnamefont {S.}~\bibnamefont
  {Mansfield}}\ and\ \bibinfo {author} {\bibfnamefont {E.}~\bibnamefont
  {Kashefi}},\ }\href {https://doi.org/10.1103/PhysRevLett.121.230401}
  {\bibfield  {journal} {\bibinfo  {journal} {Phys. Rev. Lett.}\ }\textbf
  {\bibinfo {volume} {121}},\ \bibinfo {pages} {230401} (\bibinfo {year}
  {2018})}\BibitemShut {NoStop}%
\bibitem [{\citenamefont {Okay}\ \emph {et~al.}(2018)\citenamefont {Okay},
  \citenamefont {Tyhurst},\ and\ \citenamefont {Raussendorf}}]{okay18a}%
  \BibitemOpen
  \bibfield  {author} {\bibinfo {author} {\bibfnamefont {C.}~\bibnamefont
  {Okay}}, \bibinfo {author} {\bibfnamefont {E.}~\bibnamefont {Tyhurst}}, \
  and\ \bibinfo {author} {\bibfnamefont {R.}~\bibnamefont {Raussendorf}},\
  }\href {https://doi.org/10.26421/QIC18.15-16-2} {\bibfield  {journal}
  {\bibinfo  {journal} {Quantum Information and Computation}\ }\textbf
  {\bibinfo {volume} {18}},\ \bibinfo {pages} {1272} (\bibinfo {year}
  {2018})}\BibitemShut {NoStop}%
\bibitem [{\citenamefont {Raussendorf}\ \emph {et~al.}(2020)\citenamefont
  {Raussendorf}, \citenamefont {Bermejo-Vega}, \citenamefont {Tyhurst},
  \citenamefont {Okay},\ and\ \citenamefont {Zurel}}]{raussendorf19b}%
  \BibitemOpen
  \bibfield  {author} {\bibinfo {author} {\bibfnamefont {R.}~\bibnamefont
  {Raussendorf}}, \bibinfo {author} {\bibfnamefont {J.}~\bibnamefont
  {Bermejo-Vega}}, \bibinfo {author} {\bibfnamefont {E.}~\bibnamefont
  {Tyhurst}}, \bibinfo {author} {\bibfnamefont {C.}~\bibnamefont {Okay}}, \
  and\ \bibinfo {author} {\bibfnamefont {M.}~\bibnamefont {Zurel}},\ }\href
  {https://doi.org/10.1103/PhysRevA.101.012350} {\bibfield  {journal} {\bibinfo
   {journal} {Phys. Rev. A}\ }\textbf {\bibinfo {volume} {101}},\ \bibinfo
  {pages} {012350} (\bibinfo {year} {2020})}\BibitemShut {NoStop}%
\bibitem [{\citenamefont {Nielsen}\ and\ \citenamefont
  {Chuang}(2001)}]{nielsen01}%
  \BibitemOpen
  \bibfield  {author} {\bibinfo {author} {\bibfnamefont {M.~A.}\ \bibnamefont
  {Nielsen}}\ and\ \bibinfo {author} {\bibfnamefont {I.~L.}\ \bibnamefont
  {Chuang}},\ }\href {https://doi.org/10.1017/CBO9780511976667} {\emph
  {\bibinfo {title} {Quantum Computation and Quantum Information}}}\ (\bibinfo
  {publisher} {Cambridge University Press, Cambridge, UK},\ \bibinfo {year}
  {2001})\BibitemShut {NoStop}%
\bibitem [{\citenamefont {Yan}\ and\ \citenamefont {Bacon}(2012)}]{yan12a}%
  \BibitemOpen
  \bibfield  {author} {\bibinfo {author} {\bibfnamefont {J.}~\bibnamefont
  {Yan}}\ and\ \bibinfo {author} {\bibfnamefont {D.}~\bibnamefont {Bacon}},\
  }\href@noop {} {\bibfield  {journal} {\bibinfo  {journal} {arXiv preprint}\ }
  (\bibinfo {year} {2012})},\ \Eprint {http://arxiv.org/abs/1203.3906}
  {arXiv:1203.3906 [quant-ph]} \BibitemShut {NoStop}%
\bibitem [{\citenamefont {Izmaylov}\ \emph {et~al.}(2020)\citenamefont
  {Izmaylov}, \citenamefont {Yen}, \citenamefont {Lang},\ and\ \citenamefont
  {Verteletskyi}}]{izmaylov19a}%
  \BibitemOpen
  \bibfield  {author} {\bibinfo {author} {\bibfnamefont {A.~F.}\ \bibnamefont
  {Izmaylov}}, \bibinfo {author} {\bibfnamefont {T.-C.}\ \bibnamefont {Yen}},
  \bibinfo {author} {\bibfnamefont {R.~A.}\ \bibnamefont {Lang}}, \ and\
  \bibinfo {author} {\bibfnamefont {V.}~\bibnamefont {Verteletskyi}},\ }\href
  {https://doi.org/10.1021/acs.jctc.9b00791} {\bibfield  {journal} {\bibinfo
  {journal} {Journal of Chemical Theory and Computation}\ }\textbf {\bibinfo
  {volume} {16}},\ \bibinfo {pages} {190} (\bibinfo {year} {2020})}\BibitemShut
  {NoStop}%
\bibitem [{\citenamefont {Zhao}\ \emph {et~al.}(2020)\citenamefont {Zhao},
  \citenamefont {Tranter}, \citenamefont {Kirby}, \citenamefont {Ung},
  \citenamefont {Miyake},\ and\ \citenamefont {Love}}]{zhao20a}%
  \BibitemOpen
  \bibfield  {author} {\bibinfo {author} {\bibfnamefont {A.}~\bibnamefont
  {Zhao}}, \bibinfo {author} {\bibfnamefont {A.}~\bibnamefont {Tranter}},
  \bibinfo {author} {\bibfnamefont {W.~M.}\ \bibnamefont {Kirby}}, \bibinfo
  {author} {\bibfnamefont {S.~F.}\ \bibnamefont {Ung}}, \bibinfo {author}
  {\bibfnamefont {A.}~\bibnamefont {Miyake}}, \ and\ \bibinfo {author}
  {\bibfnamefont {P.~J.}\ \bibnamefont {Love}},\ }\href
  {https://doi.org/10.1103/PhysRevA.101.062322} {\bibfield  {journal} {\bibinfo
   {journal} {Phys. Rev. A}\ }\textbf {\bibinfo {volume} {101}},\ \bibinfo
  {pages} {062322} (\bibinfo {year} {2020})}\BibitemShut {NoStop}%
\bibitem [{\citenamefont {Jordan}\ and\ \citenamefont
  {Wigner}(1928)}]{jordan28a}%
  \BibitemOpen
  \bibfield  {author} {\bibinfo {author} {\bibfnamefont {P.}~\bibnamefont
  {Jordan}}\ and\ \bibinfo {author} {\bibfnamefont {E.}~\bibnamefont
  {Wigner}},\ }\href {https://doi.org/10.1007/BF01331938} {\bibfield  {journal}
  {\bibinfo  {journal} {Z. Phys.}\ }\textbf {\bibinfo {volume} {47}},\ \bibinfo
  {pages} {631} (\bibinfo {year} {1928})}\BibitemShut {NoStop}%
\bibitem [{\citenamefont {Bravyi}\ \emph {et~al.}(2017)\citenamefont {Bravyi},
  \citenamefont {Gambetta}, \citenamefont {Mezzacapo},\ and\ \citenamefont
  {Temme}}]{bravyi17a}%
  \BibitemOpen
  \bibfield  {author} {\bibinfo {author} {\bibfnamefont {S.}~\bibnamefont
  {Bravyi}}, \bibinfo {author} {\bibfnamefont {J.~M.}\ \bibnamefont
  {Gambetta}}, \bibinfo {author} {\bibfnamefont {A.}~\bibnamefont {Mezzacapo}},
  \ and\ \bibinfo {author} {\bibfnamefont {K.}~\bibnamefont {Temme}},\
  }\href@noop {} {\bibfield  {journal} {\bibinfo  {journal} {arXiv preprint}\ }
  (\bibinfo {year} {2017})},\ \Eprint {http://arxiv.org/abs/1701.08213}
  {arXiv:1701.08213 [quant-ph]} \BibitemShut {NoStop}%
\bibitem [{\citenamefont {Setia}\ \emph {et~al.}(2020)\citenamefont {Setia},
  \citenamefont {Chen}, \citenamefont {Rice}, \citenamefont {Mezzacapo},
  \citenamefont {Pistoia},\ and\ \citenamefont {Whitfield}}]{setia20a}%
  \BibitemOpen
  \bibfield  {author} {\bibinfo {author} {\bibfnamefont {K.}~\bibnamefont
  {Setia}}, \bibinfo {author} {\bibfnamefont {R.}~\bibnamefont {Chen}},
  \bibinfo {author} {\bibfnamefont {J.~E.}\ \bibnamefont {Rice}}, \bibinfo
  {author} {\bibfnamefont {A.}~\bibnamefont {Mezzacapo}}, \bibinfo {author}
  {\bibfnamefont {M.}~\bibnamefont {Pistoia}}, \ and\ \bibinfo {author}
  {\bibfnamefont {J.~D.}\ \bibnamefont {Whitfield}},\ }\href
  {https://doi.org/10.1021/acs.jctc.0c00113} {\bibfield  {journal} {\bibinfo
  {journal} {Journal of Chemical Theory and Computation}\ }\textbf {\bibinfo
  {volume} {16}},\ \bibinfo {pages} {6091} (\bibinfo {year}
  {2020})}\BibitemShut {NoStop}%
\bibitem [{\citenamefont {Yung}\ \emph {et~al.}(2014)\citenamefont {Yung},
  \citenamefont {Casanova}, \citenamefont {Mezzacapo}, \citenamefont {McClean},
  \citenamefont {Lamata}, \citenamefont {Aspuru-Guzik},\ and\ \citenamefont
  {Solano}}]{yung14a}%
  \BibitemOpen
  \bibfield  {author} {\bibinfo {author} {\bibfnamefont {M.~H.}\ \bibnamefont
  {Yung}}, \bibinfo {author} {\bibfnamefont {J.}~\bibnamefont {Casanova}},
  \bibinfo {author} {\bibfnamefont {A.}~\bibnamefont {Mezzacapo}}, \bibinfo
  {author} {\bibfnamefont {J.}~\bibnamefont {McClean}}, \bibinfo {author}
  {\bibfnamefont {L.}~\bibnamefont {Lamata}}, \bibinfo {author} {\bibfnamefont
  {A.}~\bibnamefont {Aspuru-Guzik}}, \ and\ \bibinfo {author} {\bibfnamefont
  {E.}~\bibnamefont {Solano}},\ }\href {https://doi.org/10.1038/srep03589}
  {\bibfield  {journal} {\bibinfo  {journal} {Scientific Reports}\ }\textbf
  {\bibinfo {volume} {4}},\ \bibinfo {pages} {3589 EP } (\bibinfo {year}
  {2014})}\BibitemShut {NoStop}%
\bibitem [{\citenamefont {McClean}\ \emph {et~al.}(2016)\citenamefont
  {McClean}, \citenamefont {Romero}, \citenamefont {Babbush},\ and\
  \citenamefont {Aspuru-Guzik}}]{mcclean16a}%
  \BibitemOpen
  \bibfield  {author} {\bibinfo {author} {\bibfnamefont {J.~R.}\ \bibnamefont
  {McClean}}, \bibinfo {author} {\bibfnamefont {J.}~\bibnamefont {Romero}},
  \bibinfo {author} {\bibfnamefont {R.}~\bibnamefont {Babbush}}, \ and\
  \bibinfo {author} {\bibfnamefont {A.}~\bibnamefont {Aspuru-Guzik}},\ }\href
  {https://doi.org/10.1088/1367-2630/18/2/023023} {\bibfield  {journal}
  {\bibinfo  {journal} {New Journal of Physics}\ }\textbf {\bibinfo {volume}
  {18}},\ \bibinfo {pages} {023023} (\bibinfo {year} {2016})}\BibitemShut
  {NoStop}%
\bibitem [{\citenamefont {Romero}\ \emph {et~al.}(2018)\citenamefont {Romero},
  \citenamefont {Babbush}, \citenamefont {McClean}, \citenamefont {Hempel},
  \citenamefont {Love},\ and\ \citenamefont {Aspuru-Guzik}}]{romero18a}%
  \BibitemOpen
  \bibfield  {author} {\bibinfo {author} {\bibfnamefont {J.}~\bibnamefont
  {Romero}}, \bibinfo {author} {\bibfnamefont {R.}~\bibnamefont {Babbush}},
  \bibinfo {author} {\bibfnamefont {J.~R.}\ \bibnamefont {McClean}}, \bibinfo
  {author} {\bibfnamefont {C.}~\bibnamefont {Hempel}}, \bibinfo {author}
  {\bibfnamefont {P.~J.}\ \bibnamefont {Love}}, \ and\ \bibinfo {author}
  {\bibfnamefont {A.}~\bibnamefont {Aspuru-Guzik}},\ }\href
  {https://doi.org/10.1088/2058-9565/aad3e4} {\bibfield  {journal} {\bibinfo
  {journal} {Quantum Science and Technology}\ }\textbf {\bibinfo {volume}
  {4}},\ \bibinfo {pages} {014008} (\bibinfo {year} {2018})}\BibitemShut
  {NoStop}%
\bibitem [{\citenamefont {McClean}\ \emph {et~al.}(2018)\citenamefont
  {McClean}, \citenamefont {Boixo}, \citenamefont {Smelyanskiy}, \citenamefont
  {Babbush},\ and\ \citenamefont {Neven}}]{mcclean18a}%
  \BibitemOpen
  \bibfield  {author} {\bibinfo {author} {\bibfnamefont {J.~R.}\ \bibnamefont
  {McClean}}, \bibinfo {author} {\bibfnamefont {S.}~\bibnamefont {Boixo}},
  \bibinfo {author} {\bibfnamefont {V.~N.}\ \bibnamefont {Smelyanskiy}},
  \bibinfo {author} {\bibfnamefont {R.}~\bibnamefont {Babbush}}, \ and\
  \bibinfo {author} {\bibfnamefont {H.}~\bibnamefont {Neven}},\ }\href
  {https://doi.org/10.1038/s41467-018-07090-4} {\bibfield  {journal} {\bibinfo
  {journal} {Nature Communications}\ }\textbf {\bibinfo {volume} {9}},\
  \bibinfo {pages} {4812} (\bibinfo {year} {2018})}\BibitemShut {NoStop}%
\bibitem [{\citenamefont {Uvarov}\ and\ \citenamefont
  {Biamonte}(2021)}]{uvarov21a}%
  \BibitemOpen
  \bibfield  {author} {\bibinfo {author} {\bibfnamefont {A.}~\bibnamefont
  {Uvarov}}\ and\ \bibinfo {author} {\bibfnamefont {J.}~\bibnamefont
  {Biamonte}},\ }\href {https://doi.org/10.1088/1751-8121/abfac7} {\bibfield
  {journal} {\bibinfo  {journal} {Journal of Physics A: Mathematical and
  Theoretical}\ } (\bibinfo {year} {2021})}\BibitemShut {NoStop}%
\bibitem [{\citenamefont {Cerezo}\ \emph {et~al.}(2021)\citenamefont {Cerezo},
  \citenamefont {Sone}, \citenamefont {Volkoff}, \citenamefont {Cincio},\ and\
  \citenamefont {Coles}}]{cerezo21a}%
  \BibitemOpen
  \bibfield  {author} {\bibinfo {author} {\bibfnamefont {M.}~\bibnamefont
  {Cerezo}}, \bibinfo {author} {\bibfnamefont {A.}~\bibnamefont {Sone}},
  \bibinfo {author} {\bibfnamefont {T.}~\bibnamefont {Volkoff}}, \bibinfo
  {author} {\bibfnamefont {L.}~\bibnamefont {Cincio}}, \ and\ \bibinfo {author}
  {\bibfnamefont {P.~J.}\ \bibnamefont {Coles}},\ }\href
  {https://doi.org/10.1038/s41467-021-21728-w} {\bibfield  {journal} {\bibinfo
  {journal} {Nature Communications}\ }\textbf {\bibinfo {volume} {12}},\
  \bibinfo {pages} {1791} (\bibinfo {year} {2021})}\BibitemShut {NoStop}%
\bibitem [{\citenamefont {Wang}\ \emph {et~al.}(2020)\citenamefont {Wang},
  \citenamefont {Fontana}, \citenamefont {Cerezo}, \citenamefont {Sharma},
  \citenamefont {Sone}, \citenamefont {Cincio},\ and\ \citenamefont
  {Coles}}]{wang20a}%
  \BibitemOpen
  \bibfield  {author} {\bibinfo {author} {\bibfnamefont {S.}~\bibnamefont
  {Wang}}, \bibinfo {author} {\bibfnamefont {E.}~\bibnamefont {Fontana}},
  \bibinfo {author} {\bibfnamefont {M.}~\bibnamefont {Cerezo}}, \bibinfo
  {author} {\bibfnamefont {K.}~\bibnamefont {Sharma}}, \bibinfo {author}
  {\bibfnamefont {A.}~\bibnamefont {Sone}}, \bibinfo {author} {\bibfnamefont
  {L.}~\bibnamefont {Cincio}}, \ and\ \bibinfo {author} {\bibfnamefont {P.~J.}\
  \bibnamefont {Coles}},\ }\href@noop {} {\bibfield  {journal} {\bibinfo
  {journal} {arXiv preprint}\ } (\bibinfo {year} {2020})},\ \Eprint
  {http://arxiv.org/abs/2007.14384} {arXiv:2007.14384 [quant-ph]} \BibitemShut
  {NoStop}%
\bibitem [{\citenamefont {Motta}\ \emph {et~al.}(2020)\citenamefont {Motta},
  \citenamefont {Sun}, \citenamefont {Tan}, \citenamefont {O'Rourke},
  \citenamefont {Ye}, \citenamefont {Minnich}, \citenamefont {Brand{\~a}o},\
  and\ \citenamefont {Chan}}]{motta20a}%
  \BibitemOpen
  \bibfield  {author} {\bibinfo {author} {\bibfnamefont {M.}~\bibnamefont
  {Motta}}, \bibinfo {author} {\bibfnamefont {C.}~\bibnamefont {Sun}}, \bibinfo
  {author} {\bibfnamefont {A.~T.~K.}\ \bibnamefont {Tan}}, \bibinfo {author}
  {\bibfnamefont {M.~J.}\ \bibnamefont {O'Rourke}}, \bibinfo {author}
  {\bibfnamefont {E.}~\bibnamefont {Ye}}, \bibinfo {author} {\bibfnamefont
  {A.~J.}\ \bibnamefont {Minnich}}, \bibinfo {author} {\bibfnamefont {F.~G.
  S.~L.}\ \bibnamefont {Brand{\~a}o}}, \ and\ \bibinfo {author} {\bibfnamefont
  {G.~K.-L.}\ \bibnamefont {Chan}},\ }\href
  {https://doi.org/10.1038/s41567-019-0704-4} {\bibfield  {journal} {\bibinfo
  {journal} {Nature Physics}\ }\textbf {\bibinfo {volume} {16}},\ \bibinfo
  {pages} {205} (\bibinfo {year} {2020})}\BibitemShut {NoStop}%
\bibitem [{\citenamefont {McArdle}\ \emph {et~al.}(2019)\citenamefont
  {McArdle}, \citenamefont {Jones}, \citenamefont {Endo}, \citenamefont {Li},
  \citenamefont {Benjamin},\ and\ \citenamefont {Yuan}}]{mcardle19a}%
  \BibitemOpen
  \bibfield  {author} {\bibinfo {author} {\bibfnamefont {S.}~\bibnamefont
  {McArdle}}, \bibinfo {author} {\bibfnamefont {T.}~\bibnamefont {Jones}},
  \bibinfo {author} {\bibfnamefont {S.}~\bibnamefont {Endo}}, \bibinfo {author}
  {\bibfnamefont {Y.}~\bibnamefont {Li}}, \bibinfo {author} {\bibfnamefont
  {S.~C.}\ \bibnamefont {Benjamin}}, \ and\ \bibinfo {author} {\bibfnamefont
  {X.}~\bibnamefont {Yuan}},\ }\href
  {https://doi.org/10.1038/s41534-019-0187-2} {\bibfield  {journal} {\bibinfo
  {journal} {npj Quantum Information}\ }\textbf {\bibinfo {volume} {5}},\
  \bibinfo {pages} {75} (\bibinfo {year} {2019})}\BibitemShut {NoStop}%
\bibitem [{\citenamefont {Verteletskyi}\ \emph {et~al.}(2020)\citenamefont
  {Verteletskyi}, \citenamefont {Yen},\ and\ \citenamefont
  {Izmaylov}}]{verteletskyi20a}%
  \BibitemOpen
  \bibfield  {author} {\bibinfo {author} {\bibfnamefont {V.}~\bibnamefont
  {Verteletskyi}}, \bibinfo {author} {\bibfnamefont {T.-C.}\ \bibnamefont
  {Yen}}, \ and\ \bibinfo {author} {\bibfnamefont {A.~F.}\ \bibnamefont
  {Izmaylov}},\ }\href {https://doi.org/10.1063/1.5141458} {\bibfield
  {journal} {\bibinfo  {journal} {The Journal of Chemical Physics}\ }\textbf
  {\bibinfo {volume} {152}},\ \bibinfo {pages} {124114} (\bibinfo {year}
  {2020})}\BibitemShut {NoStop}%
\bibitem [{\citenamefont {Yen}\ \emph {et~al.}(2020)\citenamefont {Yen},
  \citenamefont {Verteletskyi},\ and\ \citenamefont {Izmaylov}}]{yen20a}%
  \BibitemOpen
  \bibfield  {author} {\bibinfo {author} {\bibfnamefont {T.-C.}\ \bibnamefont
  {Yen}}, \bibinfo {author} {\bibfnamefont {V.}~\bibnamefont {Verteletskyi}}, \
  and\ \bibinfo {author} {\bibfnamefont {A.~F.}\ \bibnamefont {Izmaylov}},\
  }\href {https://doi.org/10.1021/acs.jctc.0c00008} {\bibfield  {journal}
  {\bibinfo  {journal} {Journal of Chemical Theory and Computation}\ }\textbf
  {\bibinfo {volume} {16}},\ \bibinfo {pages} {2400} (\bibinfo {year}
  {2020})}\BibitemShut {NoStop}%
\bibitem [{\citenamefont {Gokhale}\ \emph {et~al.}(2019)\citenamefont
  {Gokhale}, \citenamefont {Angiuli}, \citenamefont {Ding}, \citenamefont
  {Gui}, \citenamefont {Tomesh}, \citenamefont {Suchara}, \citenamefont
  {Martonosi},\ and\ \citenamefont {Chong}}]{gokhale19a}%
  \BibitemOpen
  \bibfield  {author} {\bibinfo {author} {\bibfnamefont {P.}~\bibnamefont
  {Gokhale}}, \bibinfo {author} {\bibfnamefont {O.}~\bibnamefont {Angiuli}},
  \bibinfo {author} {\bibfnamefont {Y.}~\bibnamefont {Ding}}, \bibinfo {author}
  {\bibfnamefont {K.}~\bibnamefont {Gui}}, \bibinfo {author} {\bibfnamefont
  {T.}~\bibnamefont {Tomesh}}, \bibinfo {author} {\bibfnamefont
  {M.}~\bibnamefont {Suchara}}, \bibinfo {author} {\bibfnamefont
  {M.}~\bibnamefont {Martonosi}}, \ and\ \bibinfo {author} {\bibfnamefont
  {F.~T.}\ \bibnamefont {Chong}},\ }\href@noop {} {\bibfield  {journal}
  {\bibinfo  {journal} {arXiv preprint}\ } (\bibinfo {year} {2019})},\ \Eprint
  {http://arxiv.org/abs/1907.13623} {arXiv:1907.13623 [quant-ph]} \BibitemShut
  {NoStop}%
\end{thebibliography}%

\appendix

\section{Proofs}
\label[appendix]{proofs}

We will use Lemma 1 from \cite{kirby20a}:
\begin{lemma}[Lemma 1 in \cite{kirby20a}]
    \label{blochball}
    Let $P_1,P_2,...,P_N$ be an anticommuting set of Pauli operators.
    For any unit vector $\vec{a}\in\mathbb{R}^N$, the operator $\sum_{i=1}^Na_iP_i$ has eigenvalues $\pm1$.
    From this it follows that for any state, $\sum_{i=1}^N\langle P_i\rangle^2\le1$.
\end{lemma}

\noindent
\textbf{Theorem \ref{null_exp_thm}}
\emph{
    Let $\allterms$ be a set of Pauli operators, and let $\nonconterms$ be a noncontextual subset that is closed under inference within $\allterms$ (see \cref{closure_under_inference_within}).
    Then for any noncontextual state $(\vec{q},\vec{r})$ as in \eqref{noncon_params_def} describing $\nonconterms$, there exists a quantum state consistent with $(\vec{q},\vec{r})$ (i.e., that gives the same expectation values for $\nonconterms$ as $(\vec{q},\vec{r})$) for which the expectation value of every operator in $\extraterms\equiv\allterms\setminus\nonconterms$ is zero.
}
\begin{proof}
    Let $G\cup\{\mathcal{A}\}$ be the independent, commuting set of observables associated to the noncontextual state $(\vec{q},\vec{r})$ describing $\nonconterms$ (see \eqref{noncon_params_def} and \eqref{noncon_generators}): the values assigned to $G\cup\{\mathcal{A}\}$ in the noncontextual state $(\vec{q},\vec{r})$ are
    \begin{equation}
    	G_j~\mapsto~q_j=\pm1
    \end{equation}
    for each $G_j\in G$, and
    \begin{equation}
    	\mathcal{A}\equiv\sum_{i=1}^Nr_iA_i~\mapsto~+1
    \end{equation}
    (see \eqref{A_def} and the associated discussion).
    Let $P$ be a Pauli operator in $\mathcal{S}\setminus\nonconterms$.
    
    \textbf{Case 1.} If $P$ anticommutes with any operator in $G$, then $\langle P\rangle=0$, since any quantum state consistent with $(\vec{q},\vec{r})$ is a simultaneous eigenstate of $G$.
    
    \textbf{Case 2.} If $P$ commutes with the operators in $G$ and also with the $A_i$, then:
    \begin{enumerate}
        \item if $P$ is a product of operators in $G$, then $P$ can be inferred from $G$, so $P$ must in fact be included in $\nonconterms$, since by assumption $\nonconterms$ is closed under inference within $\mathcal{S}$. This follows immediately from \cref{closure_under_inference_within}, the definition of closure under inference with $\mathcal{S}$.
        \item if $P$ is not a product of operators in $G$, then $P$ is unconstrained by the noncontextual state, and may take any expectation value, including zero.
    \end{enumerate}
    
    \textbf{Case 3.} Finally, suppose $P$ commutes with the operators in $G$, but anticommutes with at least one of the $A_i$.
    In this case we want to prove that there exists a $+1$-eigenstate of $\mathcal{A}$ for which $\langle P\rangle=0$, as follows:
    
    Let $I_P$ be the set of indices such that
    \begin{equation}
    \begin{split}
        i\in I_P\quad&\Rightarrow\quad[P,A_i]=0,\\
        i\notin I_P\quad&\Rightarrow\quad\{P,A_i\}=0.
    \end{split}
    \end{equation}
    If $I_P$ is empty, then $P$ anticommutes with all of the $A_i$: thus since
    \begin{equation}
        \sum_{i=1}^N\langle A_i\rangle^2=\sum_{i=1}^Nr_i^2=1
    \end{equation}
    (see \eqref{A_def} and the associated discussion), and
    \begin{equation}
        \langle P\rangle^2+\sum_{i=1}^N\langle A_i\rangle^2\le1
    \end{equation}
    (by \cref{blochball}), it follows that $\langle P\rangle=0$.
    
    The remaining case is when $I_P$ is nonempty; there also exist $i\notin I_P$ by assumption.
    Let
    \begin{equation}
        K\equiv\sum_{i\in I_P}r_iA_i,\quad L\equiv\sum_{i\notin I_P}r_iA_i;
    \end{equation}
    thus
    \begin{equation}
        \mathcal{A}=K+L
    \end{equation}
    and
    \begin{equation}
        \quad[K,P]=0,\quad\{L,P\}=0,\quad\{K,L\}=0.
    \end{equation}
    Since $K$ and $L$ are linear combinations of anticommuting Pauli operators, their eigenvalues are $\pm k$ and $\pm l$, respectively, where
    \begin{equation}
        k\equiv\sqrt{\sum_{i\in I_P}r_i^2},\quad l\equiv\sqrt{\sum_{i\notin I_P}r_i^2}.
    \end{equation}
    Therefore,
    \begin{equation}
    \label{quasi_self_inv}
        K^2=k^2\mathds{1},\quad L^2=l^2\mathds{1},\quad k^2+l^2=1.
    \end{equation}
    
    Since $P$ commutes with $K$ and is a Pauli operator, $PK$ is also an observable with eigenvalues $\pm k$, which commutes with $L$ (since both $P$ and $K$ anticommute with $L$).
    Thus, $PK$ commutes with $\mathcal{A}$, so within the $+1$-eigenspace of $\mathcal{A}$ there exist eigenstates $|\pm\rangle$ of $PK$ with eigenvalues $\pm k$, i.e.,
    \begin{equation}
    \label{eig_vals}
        \mathcal{A}|\pm\rangle=|\pm\rangle,\quad PK|\pm\rangle=\pm k|\pm\rangle.
    \end{equation}
    Note that since $P\notin\nonconterms$, $P$ cannot be written as a product of operators in $G$ with any of the $A_i$, so both of the states $|\pm\rangle$ are consistent with the noncontextual state $(\vec{q},\vec{r})$.
    
    The noncontextual state gives us the expectation values of the $A_i$:
    \begin{equation}
        \langle\pm|A_i|\pm\rangle=r_i
    \end{equation}
    (see \eqref{A_def} and the corresponding discussion).
    This means that in addition to \eqref{eig_vals}, we have
    \begin{equation}
    \label{exp_val}
        \langle\pm|K|\pm\rangle=\sum_{i\in I_P}r_i\langle\pm|A_i|\pm\rangle=\sum_{i\in I_P}r_i^2=k^2.
    \end{equation}
    
    Define $|\psi\rangle\equiv\frac{1}{\sqrt{2}}\Big(|+\rangle+|-\rangle\Big)$; then
    \begin{equation}
    \label{one_sign}
    \begin{split}
        \langle\psi|P|\psi\rangle&=\frac{1}{k^2}\langle\psi|PK^2|\psi\rangle\quad\text{by \eqref{quasi_self_inv}}\\
        &=\frac{1}{2k^2}\Big(\langle+|+\langle-|\Big)K\Big(PK|+\rangle+PK|-\rangle\Big)\\
        &=\frac{1}{2k}\Big(\langle+|+\langle-|\Big)K\Big(|+\rangle-|-\rangle\Big)\quad\text{by \eqref{eig_vals}}\\
        &=\frac{1}{2k}\Big(\langle+|K|+\rangle-\langle+|K|-\rangle+\langle-|K|+\rangle-\langle-|K|-\rangle\Big)\\
        &=-\frac{1}{2k}\Big(\langle+|K|-\rangle-\langle-|K|+\rangle\Big)\quad\text{by \eqref{exp_val}}.
    \end{split}
    \end{equation}
    But since $K$ and $P$ commute, we similarly have
    \begin{equation}
    \label{other_sign}
    \begin{split}
        \langle\psi|P|\psi\rangle&=\frac{1}{k^2}\langle\psi|PK^2|\psi\rangle\quad\text{by \eqref{quasi_self_inv}}\\
        &=\frac{1}{2k^2}\Big(\langle+|KP+\langle-|KP\Big)K\Big(|+\rangle+|-\rangle\Big)\\
        &=\frac{1}{2k}\Big(\langle+|-\langle-|\Big)K\Big(|+\rangle+|-\rangle\Big)\quad\text{by \eqref{eig_vals}}\\
        &=\frac{1}{2k}\Big(\langle+|K|+\rangle+\langle+|K|-\rangle-\langle-|K|+\rangle-\langle-|K|-\rangle\Big)\\
        &=\frac{1}{2k}\Big(\langle+|K|-\rangle-\langle-|K|+\rangle\Big)\quad\text{by \eqref{exp_val}}.
    \end{split}
    \end{equation}
    Together, \eqref{one_sign} and \eqref{other_sign} imply that
    \begin{equation}
        \langle\psi|P|\psi\rangle=-\langle\psi|P|\psi\rangle=0,
    \end{equation}
    so since $|\psi\rangle$ is a state in the $+1$-eigenspace of $\mathcal{A}$, we are done.
\end{proof}
~
\noindent
\begin{lemma}
    \label{diagonalizecommutinglemma}
    For any set of $M$ independent, commuting Pauli operators, there exists an efficiently calculable unitary rotation $U$, given by a sequence of at most $2M$ $\frac{\pi}{2}$-rotations generated by Pauli operators, that maps the set to a set of distinct single-qubit $Z$ operators.
\end{lemma}
\begin{proof}
Let $\{B_i~|~i=1,2,...,M\}$ be a commuting set of Pauli operators.
We may write $B_i$ as
\begin{equation}
	\label{Aexpansion}
	B_i=\bigotimes_{k=1}^N\sigma^{(B_i)}_k,
\end{equation}
where each $\sigma^{(B_i)}_k\in\{I,X,Y,Z\}$.
We consider two cases:

\textbf{Case 1.} Suppose $B_i$ is not diagonal, meaning that there is some $k$ such that $\sigma^{(B_i)}_k\in\{X,Y\}$.
Consider the Pauli operator
\begin{equation}
	J_i=\bigotimes_{k=1}^N\sigma^{(J_i)}_k,
\end{equation}
defined as follows: for each $k$...
\begin{align}
	\sigma^{(B_i)}_k=I~&\Rightarrow~\sigma^{(J_i)}_k=I,\\
	\sigma^{(B_i)}_k=Z~&\Rightarrow~\sigma^{(J_i)}_k=Z,\\
	\sigma^{(B_i)}_k=X\text{ or }Y~&\Rightarrow~\sigma^{(J_i)}_k=X\text{ or }Y,\label{last_line}
\end{align}
where the values $\sigma^{(J_i)}_k=X\text{ or }Y$ in \eqref{last_line} are chosen so that $\sigma^{(B_i)}_k$ and $\sigma^{(J_i)}_k$ differ for exactly one $k$ (as we noted above, at least one of the $\sigma^{(B_i)}_k$ is $X$ or $Y$ if $B_i$ is not diagonal.)
This guarantees that $J_i$ anticommutes with $B_i$.

Consider a rotation by $\pi/2$ generated by $J_i$, i.e.,
\begin{equation}
	\exp\left(i\frac{\pi}{4}J_i\right)=\frac{1}{\sqrt{2}}(1+iJ_i).
\end{equation}
Upon conjugating $B_i$ by this operator, we obtain
\begin{equation}
\label{diagonalizer}
\begin{split}
	&\frac{1}{2}(1+iJ_i)B_i(1-iJ_i)\\
	&=\frac{1}{2}(B_i-iB_iJ_i+iJ_iB_i+J_iB_iJ_i)\\
	&=\frac{1}{2}(B_i+2iJ_iB_i-J_iJ_iB_i)\\
	&=iJ_iB_i\\
	&=i\bigotimes_{k=1}^N\sigma^{(J_i)}_k\sigma^{(B_i)}_k,
\end{split}
\end{equation}
where the third line follows because $J_i$ anticommutes with $B_i$, and the fourth line follows because $J_i$ is self-inverse.
By the conditions on the $\sigma^{(J_i)}_k$, we see that $\sigma^{(J_i)}_k\sigma^{(B_i)}_k=I\text{ or }\pm iZ$ for each $k$, and $\pm iZ$ appears exactly once, so \eqref{diagonalizer} becomes
\begin{equation}
	\frac{1}{2}(1+iJ_i)B_i(1-iJ_i)=\pm\bigotimes_{k=1}^N\sigma^{(D_i)}_k,
\end{equation}
where all $\sigma^{(D_i)}_k=I$ except one, which is $Z$.
In other words, the rotation about $J_i$ has mapped $B_i$ to a single-qubit $Z$ operator, as desired.

In each step, we apply the rotation $\exp\left(i\frac{\pi}{4}J_i\right)$ to all operators in the set. Thus we might worry that, having already mapped some subset of the $B_{i'}$ to single-qubit $Z$ operators $D_{i'}$, applying some later rotation $\exp\left(i\frac{\pi}{4}J_i\right)$ to map $B_i$ to a single-qubit $Z$ operator could change the previously obtained $D_{i'}$.
This turns out not to be the case, as we now show:

Consider some particular one of the $D_{i'}$, whose expansion as a tensor product is
\begin{equation}
	D_{i'}=\bigotimes_{k=1}^N\sigma^{(D_{i'})}_k
\end{equation}
where one of the $\sigma^{(D_{i'})}_k$ is $Z$ and the others are $I$.
$D_{i'}$ commutes with $B_i$, since the previously applied rotations preserve commutation relations, so for all values of $m$ such that $\sigma^{(D_{i'})}_m=Z$, $\sigma^{(B_i)}_m$ (as defined in \eqref{Aexpansion}) must be $I$ or $Z$.
But this implies that $J_i$ also commutes with $D_{i'}$, since we know that $\sigma^{(J_i)}_m$ is $I$ or $Z$ exactly when $\sigma^{(B_i)}_m$ is $I$ or $Z$, and thus $\sigma^{(J_i)}_m$ is $I$ or $Z$.

Therefore, the rotation that maps $B_i$ to a single-qubit $Z$ preserves the previously obtained $D_{i'}$.

\textbf{Case 2.} Suppose $B_i$ is diagonal, so $\sigma^{(B_i)}_k\in\{I,Z\}$ for all $k$.
Since any previously-obtained $D_{i'}$ are single-qubit $Z$ operators and we assumed that the entire set is independent, $B_i$ cannot be the product of any subset of the previously-obtained $D_{i'}$.
Therefore, there must exist some $m\in\{1,2,...,n\}$ such that 
\begin{equation}
	\sigma^{(B_i)}_m=Z
\end{equation}
and
\begin{equation}
	\sigma^{(D_{i'})}_m=I
\end{equation}
for all of the previously-obtained $D_{i'}$.

Apply the rotation $\exp\left(i\frac{\pi}{4}K_i\right)$, for $K_i$ defined by
\begin{equation}
	K_i=\bigotimes_{k=1}^N\sigma^{(K_i)}_k,
\end{equation}
where
\begin{equation}
	\sigma^{(K_i)}_m=Y
\end{equation}
and
\begin{equation}
	\sigma^{(K_i)}_k=I
\end{equation}
for all $k\neq m$.
Thus $\exp\left(i\frac{\pi}{4}K_i\right)$ commutes with and therefore does not change any previously-obtained $D_{i'}$.

As in Case 1, applying $\exp\left(i\frac{\pi}{4}K_i\right)$ to $B_i$ obtains
\begin{equation}
\label{undiag}
	\exp\left(i\frac{\pi}{4}K_i\right)B_i\exp\left(-i\frac{\pi}{4}K_i\right)=i\bigotimes_{k=1}^N\sigma^{(K_i)}_k\sigma^{(B_i)}_k,
\end{equation}
where by construction,
\begin{equation}
	\sigma^{(K_i)}_m\sigma^{(B_i)}_m=iX
\end{equation}
and
\begin{equation}
	\sigma^{(K_i)}_k\sigma^{(B_i)}_k=\sigma^{(B_i)}_k
\end{equation}
for all $k\neq m$.
In other words, the rotation has changed the $Z$ at the $m$th spot in $B_i$ into an $X$, and left $B_i$ otherwise unchanged.
We also apply this rotation to all other operators in the set, which does not change those that have already been mapped to single-qubit $Z$ operators, as we noted above.

Now $B_i$ is no longer diagonal, so we proceed as described in Step 1 above, applying a second Pauli $\frac{\pi}{2}$-rotation to map $B_i$ to a single-qubit Pauli operator.
\end{proof}

\section{CS-VQE implementation details}
\label[appendix]{applications_app}

\begin{figure*}
\centering
\includegraphics[width=7in]{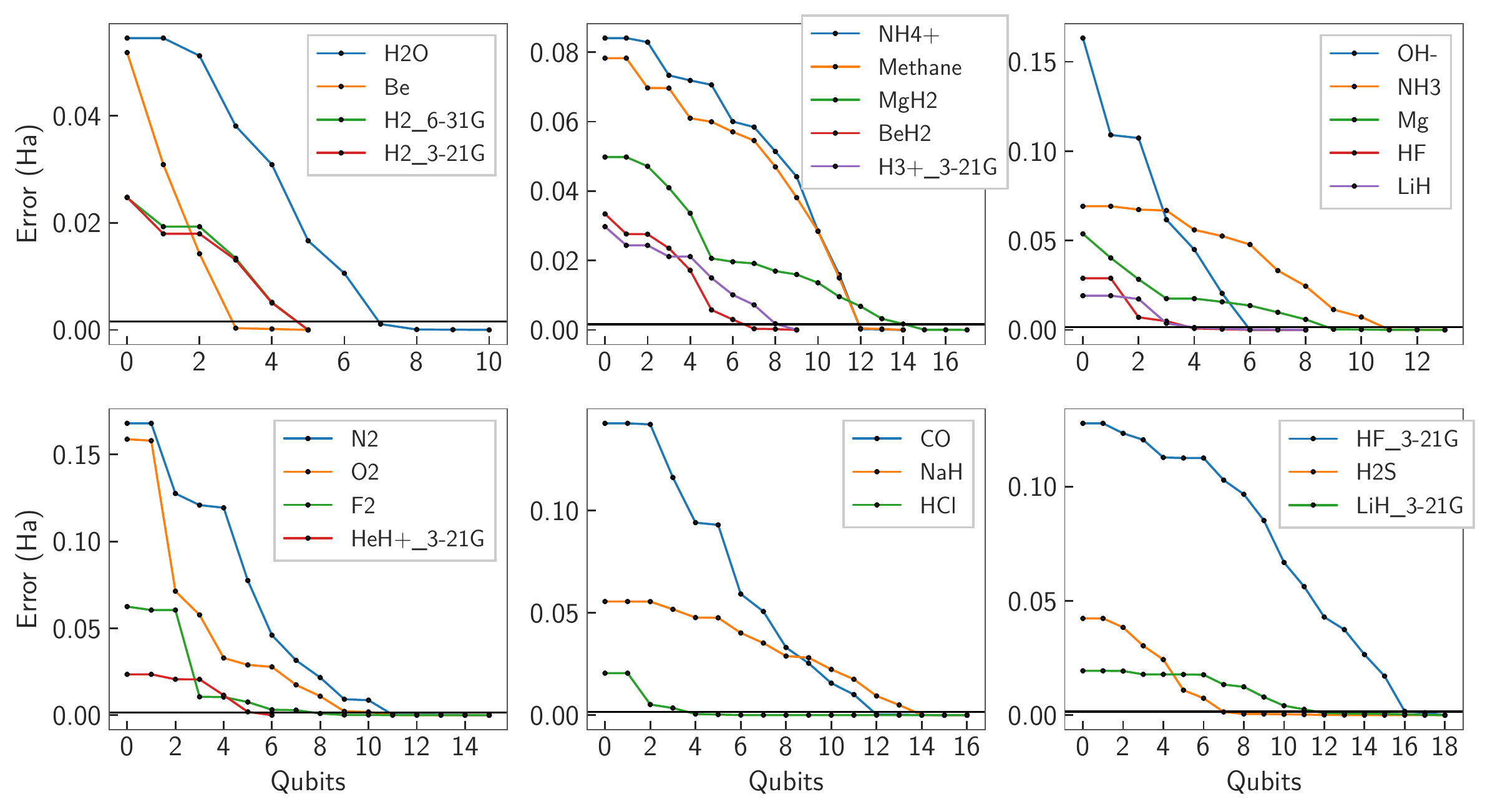}
\caption{
CS-VQE approximation errors versus number of qubits used on the quantum computer, for tapered molecular Hamiltonians, using the qubit ordering chosen by the optimal heuristic. All Hamiltonians whose curves overlap in the region below chemical accuracy have the same total numbers of qubits. The solid black lines indicate chemical accuracy. Within each subplot, the ordering of the legend matches the vertical ordering of the leftmost points in the curves.
\label{errors_vs_qubits_opt}
}
\end{figure*}

\begin{figure}
\centering
\includegraphics[width=\columnwidth]{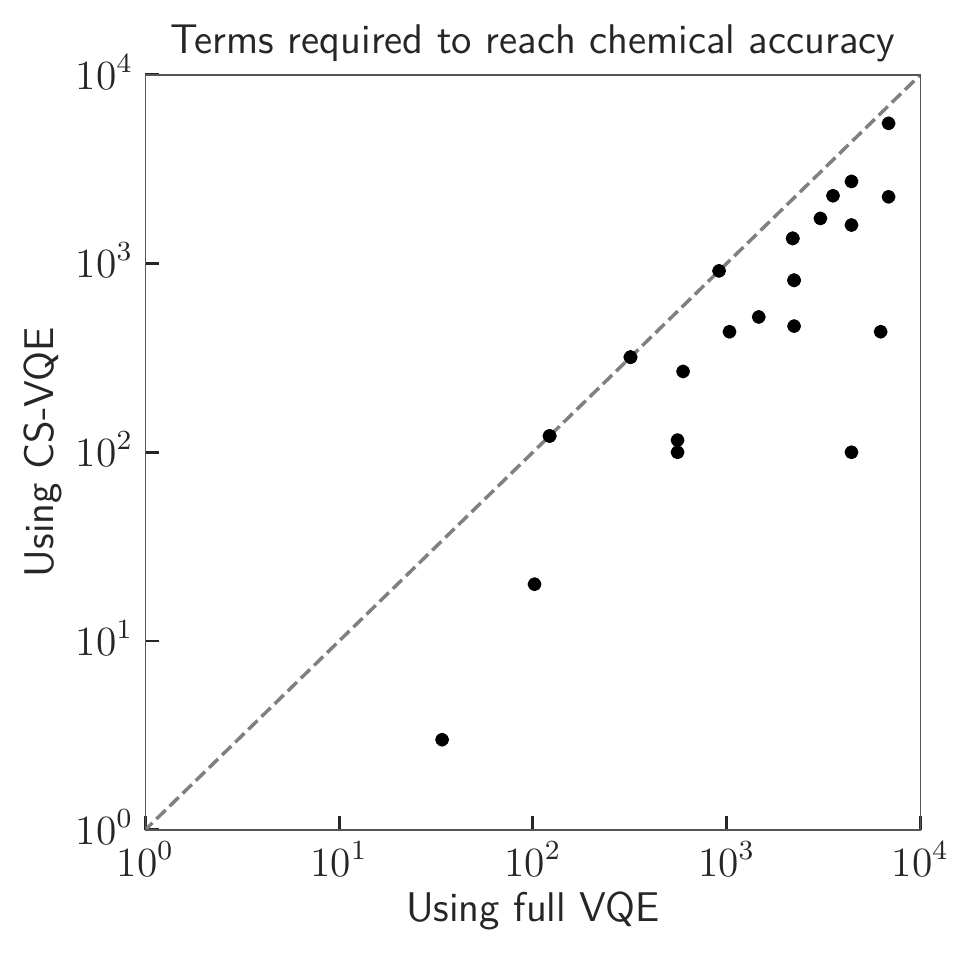}
\caption{
Number of terms simulated on the quantum processor required to reach chemical accuracy using CS-VQE versus using full VQE, for qubit ordering chosen by the optimal heuristic.
The dashed line marks equality.
All points represent either one, two, or three Hamiltonians.
\label{fig:terms_csvqe_vs_full_opt}}
\end{figure}

\subsection{Moving qubits from the noncontextual approximation to the quantum correction}

We describe in detail how the noncontextual approximation is truncated to make room for improved quantum corrections.
As in \cref{quantum_part}, we work in the rotated basis, denote by $\mathcal{H}_1$ the subspace of $n_1$ qubits acted upon by the noncontextual generators $G'$, and denote by $\mathcal{H}_2$ the subspace of the $n_2$ remaining qubits, which are used to implement the quantum correction.
Let $\mathcal{I}_2$ be the set of indices of these latter $n_2$ qubits (those not acted upon by $G'$), whose Hilbert space is the quantum search space $\mathcal{H}_2$.
To increase the size of $\mathcal{H}_2$, we first need to select some subset of $G'$ that we want to remove from $\nonconham$.
Since the elements of $G'$ are single-qubit $Z$ operators, this subset defines a set $\mathcal{I}_\text{add}$ of indices for qubits whose states are initially fixed by the noncontextual state, but that we will switch to simulating on the quantum processor.
To begin with, from $\nonconterms'$ (the terms in the noncontextual part of the Hamiltonian, in the rotated basis) we remove all terms that act on any of the qubits in $\mathcal{I}_\text{add}$, including the elements of $G'$ that act on these qubits.
The remaining elements of $G'$ form a new generating set $G''\subset G'$ satisfying
\begin{equation}
\label{new_generators}
    |G''|=|G'|-|\mathcal{I}_\text{add}|.
\end{equation}

Let the new noncontextual set of terms be denoted $\nonconterms''$, and let
\begin{equation}
    \nonconham''\equiv\sum_{P\in\nonconterms''}h_PP.
\end{equation}
All terms that were removed from $\nonconterms'$ to obtain $\nonconterms''$ should be added to $\extraterms'$ to obtain an expanded contextual set of terms $\extraterms''$, whose corresponding Hamiltonian is
\begin{equation}
    \extraham''\equiv\sum_{P\in\extraterms''}h_PP.
\end{equation}

We can see that this removal operation preserves closure under inference of $\nonconterms''$ by again thinking of $G'\cup\mathcal{A}'$ as stabilizers (up to some signs) for the contextual subspace.
The elements of $G'$ are single Pauli operators that are generators for the stabilizer group of the subspace.
In order to increase the dimension of the stabilized subspace, therefore, for each element $G_i'$ of $G'$ that we remove we must also remove all elements of the stabilizer group that include $G_i'$ as a factor.
In other words, in this instance preserving closure under inference is equivalent to preserving closure of a stabilizer group.

We can now implement the quantum correction on $\nonconham''$ and $\extraham''$, keeping the same noncontextual ground state that we began with, but only applying its value assignments to terms in $\nonconham''$.
Let $n_2''$ denote the new number of qubits used in the quantum correction procedure: then by \eqref{new_generators},
\begin{equation}
    n_2''=n-|G''|=n-|G|+|\mathcal{I}_\text{add}|=n_2+|\mathcal{I}_\text{add},
\end{equation}
where $n_2$ was the initial number of qubits used in the quantum correction procedure.

\subsection{Details of heuristics}

The heuristic described in the main text starts from the pure noncontextual approximation and moves qubits two at a time to the quantum correction search space, greedily maximizing the improvement to the error at each step.
We guessed that this heuristic might perform well for the following reasons.
For the molecular Hamiltonians we tested, the noncontextual Hamiltonians chosen via the greedy heuristic discussed in \cite{kirby20a} contain the diagonal terms in the full Hamiltonian (i.e., tensor products of combinations of Pauli $Z$ and single-qubit identity), together with a single clique containing some off-diagonal terms.
In particular, this means that the generating set $G$ comprises single-qubit $Z$ operators in the original basis, so the rotated basis is identical to the original basis.
The highest weight terms that are not included in the noncontextual part are those containing only one or two off-diagonal Pauli tensor factors.

As discussed in the main text, terms in the quantum correction Hamiltonian are ``freed'' to be optimized when the generators in the noncontextual part that they anticommute with are dropped.
In the present case, these generators are the single-qubit $Z$ operators that act on the qubits for which the terms in the quantum correction Hamiltonian are off-diagonal.
Thus, greedily dropping pairs of generators allows a new subset of the highest weight terms in the quantum correction Hamiltonian to be freed for optimization at each step in the heuristic.

For other heuristics, we refer to the one that led to the best errors of any we tested as the optimal heuristic.
The optimal heuristic begins with full VQE, then adds qubits to the noncontextual approximation one at a time while greedily minimizing the error penalty at each step.
This heuristic is consequently as hard as performing full VQE, so using it in practical implementations would negate the value of CS-VQE.
However, it is informative because it provides a good approximation to the optimal orderings and consequently to the ideal performance of CS-VQE.

For comparison to \cref{errors_vs_qubits_pract} and \cref{fig:terms_csvqe_vs_full_practical}, the errors versus qubits and terms to reach chemical accuracy figures for the heuristic in the main text, we include here the corresponding figures for the optimal heuristic, as \cref{errors_vs_qubits_opt} and \cref{fig:terms_csvqe_vs_full_opt}.
Notably, the heuristic included in the main text matches the number of qubits required to reach chemical accuracy using the optimal heuristic in all cases except F$_2$, LiH in the 3-21G basis, and Mg: in these cases the heuristic in the main text requires one more qubit than the optimal heuristic.
In fact, for HeH+ in the 3-21G basis, the heuristic in the main text requires one fewer qubit to reach chemical accuracy than the optimal heuristic.

As an alternative to heuristics that calculate actual energies, we tested two heuristics based on the total weight of the Pauli terms associated to each qubit.
Both starting from the full VQE end (greedily minimizing the penalty for each qubit removed) and starting from the noncontextual end (greedily maximizing the improvement for each qubit added) had identical performance for the examples we tested, but unfortunately that performance was substantially worse than the performance of the heuristic discussed in the main text (as well as the optimal heuristic).

\section{Noncontextual Hamiltonians}
\label[appendix]{noncontextual_part_app}

As noted in the main text, a set of observables is \emph{noncontextual} when it admits consistent joint valuations.
The kind of contradiction that might prevent such a joint valuation is closely related to the notion of \emph{inference}, which we can introduce as follows.
If a pair of observables $A,B$ commute, then they can be measured simultaneously together with their product $AB$.
Thus, if we attempt to construct a classical, ontological model for some set of observables including $A$ and $B$, in any assignment of values, the value assigned to $AB$ must be the product of the values assigned to $A$ and $B$.
This is based on the fact that if an observer measures $A$, $B$, and $AB$, the values they obtain will always be consistent with the product relation.
Hence, we say that given an assignment of values to $A$ and $B$, we may \emph{infer} the value assignment to $AB$.

\begin{definition}[see \cite{kirby19a}]
\label{closure_under_inference_def}
    Given an arbitrary set $\allterms$ of Pauli operators, the \emph{closure under inference} $\overline{\allterms}$ of $\allterms$ is the minimal set of Pauli operators, containing $\allterms$ as a subset, such that for every commuting pair $A,B$ in $\overline{\allterms}$, $AB$ is also in $\overline{\allterms}$.
\end{definition}
\begin{definition}[see \cite{kirby19a}]
\label{noncontextuality_def}
    A set $\allterms$ of Pauli operators is \emph{noncontextual} if it is possible to assign values ($\pm1$) to $\overline{\allterms}$ that respect all inference relations in $\overline{\allterms}$, i.e., such that for every commuting pair $A,B\in\overline{\allterms}$, the value assigned to $AB$ is the product of the values assigned to $A,B$.
\end{definition}

A set of Pauli operators $\nonconterms$ is noncontextual if and only if it has the form
\begin{equation}
\label{noncontextual_form}
    \nonconterms=\mathcal{Z}\cup C_1\cup C_2\cup\cdots\cup C_N,
\end{equation}
where $\mathcal{Z}$ is the subset of $\nonconterms$ containing all operators in $\nonconterms$ that commute with all other operators in $\nonconterms$, operators in the same $C_i$ commute, and operators in different $C_i$ anticommute \cite{kirby19a}.
In other words, $\nonconterms$ is noncontextual if and only if commutation is transitive on $\nonconterms\setminus\mathcal{Z}$.
If commutation is transitive on a set, then because it is always reflexive and symmetric, it is an equivalence relation: as noted in the main text, we label its equivalence classes $C_i$, and also refer to these as cliques.

As in the main text, we partition a general Hamiltonian $\ham$ into a noncontextual part $\nonconham$ and the remaining terms $\extraham$, where the associated sets of Pauli operators are $\allterms$, $\nonconterms$, and $\extraterms$, respectively.
$\nonconterms$ must be closed under inference within $\allterms$, which we may now define rigorously in terms of \cref{closure_under_inference_def}: $\nonconterms$ is closed under inference within $\allterms$ if and only if
\begin{equation}
\label{c_u_i_within}
    \nonconterms=\overline{\nonconterms}\cap\allterms.
\end{equation}
This is simply a formalization of \cref{closure_under_inference_within} in the main text.
Subject to these constraints, we can choose $\nonconham$ in any way we like.

The key step in building a quasi-quantized model for the noncontextual part $\nonconham$ is construction of $\mathcal{R}$, a new set of Pauli operators such that $\overline{\mathcal{R}}=\overline{\nonconterms}$ (so that value assignments to $\mathcal{R}$ induce value assignments to $\nonconterms$ by inference), and $\mathcal{R}$ is independent:
\begin{definition}[see \cite{kirby20a}]
A set $\mathcal{R}$ of Pauli operators is \emph{independent} if no value of any operator $A$ in $\mathcal{R}$ can be inferred from any value assignment to a subset of $\mathcal{R}$ not containing $A$.
Equivalently, $\mathcal{R}$ is independent if and only if for every commuting subset of $\mathcal{R}$, its product is not in $\mathcal{R}$.
\end{definition}
Note that in a commuting set of Pauli operators, this notion of independence reduces to the usual definition of independence of subsets of an Abelian group \cite{kirby20a}.
Requiring $\mathcal{R}$ to be independent means that not only are \emph{some} value assignments to $\mathcal{R}$ allowed (since $\mathcal{R}$ is noncontextual), but in fact \emph{all} value assignments to $\mathcal{R}$ are allowed (which is not true for a general noncontextual set) \cite{kirby20a}.

The independent set $\mathcal{R}$ is given by
\begin{equation}
    \mathcal{R}=G\cup\{A_i~|~i=1,2,...,N\}
\end{equation}
where each $A_i\in C_i$, and $G$ is an independent generating set for
\begin{equation}
\label{univ_comm_terms}
    \mathcal{Z}\cup\bigcup_{i=1}^N\{AA'~|~A,A'\in C_i\}.
\end{equation}
Note that since the set \eqref{univ_comm_terms} is composed of elements of $\mathcal{Z}$ and products of pairs of elements in the same clique, its elements commute with all elements of $\nonconterms$.
All elements of $G$ therefore commute with all of the $A_i$, although the $A_i$ pairwise anticommute (since each is an element of the corresponding $C_i$).
As noted in \eqref{noncon_generators} in the main text, states of the quasi-quantized model turn out to be equivalent to joint knowledge of the set of commuting observables
\begin{equation}
\label{noncon_generators_app}
    G\cup\{\mathcal{A}\},
\end{equation}
where the Pauli operators $G_j\in G$ have values $q_j=\pm1$, and for some unit vector $\vec{r}$ the operator 
\begin{equation}
\label{A_def_app}
    \mathcal{A}\equiv\sum_{i=1}^Nr_iA_i
\end{equation}
has value $+1$ (and consequently each operator $A_i$ has expectation value $r_i$).
Since the Pauli operators $A_i$ anticommute and $|\vec{r}|=1$, $\mathcal{A}$ has eigenvalues $\pm1$, as shown in \cite{kirby20a} (see also \cref{blochball}).
The states of the quasi-quantized model are thus parametrized by the values $(\vec{q},\vec{r})$, so we call such a parameter set a \emph{noncontextual state} (which is the same as an \emph{epistemic state}, as in \cite{spekkens16a,kirby20a}).

The resulting expression for the expectation value of $\nonconham$ is
\begin{equation}
\label{objfn}
    \langle\nonconham\rangle_{(\vec{q},\vec{r})}
    =\sum_{B\in\overline{G}}\left(h_B+\sum_{i=1}^Nh_{B,i}r_i\right)\prod_{j\in\mathcal{J}_B}q_j,
\end{equation}
where the classical state parameters $(\vec{q},\vec{r})$ can take any values such that $q_j=\pm1$ for each $j$ and $|\vec{r}|=1$ \cite{kirby20a}.
The constants in \eqref{objfn}, $h_B$ and $h_{B,i}$, are the coefficients in the original Hamiltonian $\nonconham$ (under an efficiently classically calculable relabeling), and for each $B\in\overline{G}$, $\mathcal{J}_B$ is the set of indices such that
\begin{equation}
    B=\prod_{j\in\mathcal{J}_B}G_j,
\end{equation}
which is also efficiently classically calculable.
Thus \eqref{objfn} expresses the expectation value of the noncontextual part of the Hamiltonian as a classical objective function of the parameters $(\vec{q},\vec{r})$, which may be both obtained and evaluated classically efficiently \cite{kirby20a}.

Given the objective function \eqref{objfn}, estimating the ground state energy of the noncontextual part of the Hamiltonian requires minimizing \eqref{objfn} over the parameters $(\vec{q},\vec{r})$.
For an Hamiltonian of $n$ qubits, the total dimension of $(\vec{q},\vec{r})$ is at most $2n+1$.
As noted in the main text, we refer to the setting $(\vec{q},\vec{r})$ that minimizes \eqref{objfn} as the \emph{noncontextual ground state}.

\end{document}